\documentclass[11pt,a4paper]{article}

\usepackage{fullpage}

\usepackage[latin1]{inputenc}
\usepackage[english]{babel}
\usepackage{amsmath}
\usepackage{amsfonts}
\usepackage{amssymb}
\usepackage{amsthm}
\usepackage{mathtools}
\usepackage[svgnames]{xcolor}
\usepackage{tikz}
\usepackage{sgame}
\usepackage{revsymb}
\usepackage{subcaption}
\usepackage[font={small}, width={0.9\textwidth}]{caption}
\usepackage{url}

\newtheorem{theorem}{Theorem}

\newtheorem{definition}[theorem]{Definition}

\newtheorem{fact}[theorem]{Fact}
\newtheorem{lemma}[theorem]{Lemma}

\newtheorem{proposition}[theorem]{Proposition}
\newtheorem{remark}[theorem]{Remark}

\newcommand{\eps}{\varepsilon}

\newcommand{\sv}{\mathbf{s}}
\newcommand{\rv}{\mathbf{r}}
\newcommand{\xv}{\mathbf{x}}
\newcommand{\fv}{\mathbf{f}}
\newcommand{\gv}{\mathbf{g}}
\newcommand{\av}{\mathbf{a}}
\newcommand{\bv}{\mathbf{b}}
\newcommand{\tv}{\mathbf{t}}
\newcommand{\lamv}{{\boldsymbol \lambda}}
\newcommand{\gamv}{{\boldsymbol \gamma}}
\newcommand{\sigv}{{\boldsymbol \sigma}}

\newcommand{\nc}{\newcommand}

\newcommand{\dio}[1]{#1}

\nc{\Rank}{\operatorname{Rank}}
\nc{\smfrac}[2]{\mbox{$\frac{#1}{#2}$}}
\nc{\tr}{\operatorname{Tr}}
\nc{\Tr}{\operatorname{Tr}}
\nc{\ox}{\otimes}
\nc{\dg}{\dagger}
\nc{\dn}{\downarrow}
\nc{\cA}{{\cal A}}
\nc{\cB}{{\cal B}}
\nc{\cC}{{\cal C}}
\nc{\cD}{{\cal D}}
\nc{\cE}{{\cal E}}
\nc{\cF}{{\cal F}}
\nc{\cG}{{\cal G}}
\nc{\cH}{{\cal H}}
\nc{\cI}{{\cal I}}
\nc{\cJ}{{\cal J}}
\nc{\cK}{{\cal K}}
\nc{\cL}{{\cal L}}
\nc{\cM}{{\cal M}}
\nc{\cN}{{\cal N}}
\nc{\cO}{{\cal O}}
\nc{\cP}{{\cal P}}
\nc{\cQ}{{\cal Q}}
\nc{\cR}{{\cal R}}
\nc{\cS}{{\cal S}}
\nc{\cT}{{\cal T}}
\nc{\cX}{{\cal X}}
\nc{\cY}{{\cal Y}}
\nc{\cZ}{{\cal Z}}
\nc{\csupp}{{\operatorname{csupp}}}
\nc{\qsupp}{{\operatorname{qsupp}}}
\nc{\var}{{\operatorname{var}}}
\nc{\rar}{\rightarrow}
\nc{\lrar}{\longrightarrow}
\nc{\polylog}{{\operatorname{polylog}}}
\nc{\1}{{\openone}}
\nc{\wt}{{\operatorname{wt}}}
\nc{\avg}[1]{{\left\langle {#1} \right\rangle}}

\def\a{\alpha}
\def\b{\beta}
\def\g{\gamma}
\def\d{\delta}

\nc{\RR}{{{\mathbb R}}}
\nc{\CC}{{{\mathbb C}}}
\nc{\FF}{{{\mathbb F}}}
\nc{\NN}{{{\mathbb N}}}
\nc{\ZZ}{{{\mathbb Z}}}
\nc{\PP}{{{\mathbb P}}}
\nc{\QQ}{{{\mathbb Q}}}
\nc{\UU}{{{\mathbb U}}}
\nc{\EE}{{{\mathbb E}}}
\nc{\id}{{\operatorname{\tt id}}}

\newcommand{\ket}[1]{|#1\rangle}
\newcommand{\bra}[1]{\langle#1|}

\nc{\ketbra}[2]{|#1\rangle\!\langle#2|}
\nc{\proj}[1]{| #1\rangle\!\langle #1 |}

\newcommand{\COMMEQ}{\mathtt{Comm}}
\newcommand{\BIEQ}{\mathtt{B.I.}}
\newcommand{\CORREQ}{\mathtt{Corr}}
\newcommand{\NEQ}{\mathtt{Nash}}
\newcommand{\QUEQ}{\mathtt{Quantum}}

\newcommand{\SW}{\mathtt{SW}}
\newcommand{\SPO}{\mathtt{SPO}}

\title{Belief-Invariant and Quantum Equilibria \\ in Games of Incomplete Information}

\author{Vincenzo Auletta\thanks{Dipartimento di Ingegneria dell'Informazione ed Elettrica e Matematica applicata (DIEM), Universit\`a degli Studi di Salerno, Italy. Email: {\tt auletta@unisa.it}}, Diodato Ferraioli\thanks{Dipartimento di Ingegneria dell'Informazione ed Elettrica e Matematica applicata (DIEM), Universit\`a degli Studi di Salerno, Italy. Email: {\tt dferraioli@unisa.it}}, Ashutosh Rai\thanks{Institute of Physics, Slovak Academy of Sciences, Bratislava, Slovakia. Email: {\tt ashutosh.rai@savba.sk}},\\ Giannicola Scarpa\thanks{
Escuela T\'ecnica Superior de Ingenier\'ia de Sistemas Inform\'aticos (ETSISI), Universidad Polit\'ecnica de Madrid, Spain. Email: {\tt g.scarpa@upm.es}}, Andreas Winter\thanks{ICREA and Departament de F\'{\i}sica: Grup d'Informaci\'{o} Qu\`{a}ntica, Universitat Aut\`{o}noma de Barcelona, Spain. Email: {\tt andreas.winter@uab.cat}}}

\date{\today}

\begin{document}

\raggedbottom

\maketitle

\begin{abstract}

Drawing on ideas from game theory and quantum physics, we investigate
nonlocal correlations from the point of view of equilibria in games of
incomplete information. These equilibria can be classified in decreasing
power as general \emph{communication equilibria},
\emph{belief-invariant equilibria}
and \textit{correlated equilibria},
all of which contain the familiar \emph{Nash equilibria}.

The notion of belief-invariant equilibrium appeared in game theory in the 90s. However, the class of \emph{non-signalling correlations}
associated to belief-invariance arose naturally already in the 80s
in the foundations of quantum mechanics.

In the present work, we explain and unify these two origins of the idea
and study the above classes of equilibria, together with \emph{quantum correlated
equilibria}, using tools from quantum information but the language of (algorithmic) game theory.
We present a general framework of belief-invariant communication equilibria,
which contains correlated equilibria and quantum correlated equilibria as
special cases.
Our framework also contains the theory of Bell inequalities and their violations due to
non-locality, which is a question of intense interest in the foundations of
quantum mechanics, and it was indeed the original motivation for the aforementioned
studies. Moreover, in our framework we can also model quantum games where players have conflicting interests,
a recent developing topic in physics.

We then use our framework to show new results related to the
\textit{social welfare} of equilibria. Namely,
we exhibit a game where belief-invariance is socially better than any
correlated equilibrium, and a game where all non-belief-invariant communication
equilibria have a suboptimal social welfare.
We also show that optimal social welfare can in certain cases be achieved by quantum mechanical correlations, which do not need an informed mediator to be implemented, and go beyond the classical  ``sunspot" or shared randomness approach.


\end{abstract}

\newpage

\section{Introduction}
\label{sec:intro}
The notion of equilibrium of a strategic game and the mathematical formulation of rational
behaviour in situations of conflict are among the most fruitful ideas
in the history of economics, and duly became a cornerstone of any modern discussion on the subject.

The topic was initiated by the classic
treatment of von Neumann and Morgenstern~\cite{vNeumann-Morgenstern},
where it was realized that in the realm of mixed
strategies there is always a minimax equilibrium for zero-sum games. Another milestone was the definition of
\emph{Nash equilibrium}~\cite{Nash} and Nash's proof that in mixed strategies
there always exists one. These pioneering results were followed by a multitude of further investigations
into other concepts of equilibrium and their properties, including the
question of how the players, knowing the game, can compute an equilibrium \cite{aumann1974subjectivity,Forges93,lemke1964equilibrium,daskalakis2009complexity,Chen:2009}.
Motivated, among other things, by the realization that Nash equilibria
sometimes can be ``bad'' both individually and collectively for the
players, a major direction in game theory is the question of how to induce
players, or help them, toward a more beneficial equilibrium \cite{fudenberg1991game,opac-b1123911,NisaRougTardVazi07}. One important idea in this line of research
is that giving the players an \emph{advice}, in the form of a random variable,
can change the landscape of equilibria. This
generalizes the concept of Nash equilibrium to \emph{correlated equilibria} \cite{aumann1974subjectivity}.

The present paper is about advice in the setting of \textit{games of incomplete information}.
As it turns out, this is a subject of considerable complexity,
since the kind of correlation that can serve as advice to the players
can be far more general than in the case of complete information.
In games of incomplete information, or \emph{Bayesian games},
each player has a \emph{type} which is not perfectly known to, but only estimated by, the other players.
Depending on what the game models, a type can be many things. For example, it can represent a characteristic of the player (strong, weak, rich, poor, etc.) or a secret objective of the player (interest in one particular outcome).
For this class of games, a relevant solution concept is
the \emph{communication equilibrium}~\cite{Forges82}. Here, the players privately communicate their type
to a mediator, who implements a correlation and gives each player advice for a convenient action.
It is reasonable to assume that players are comfortable with revealing their private information
to a trusted mediator if this gives them an advantage. However, there are situations where
it might be crucial for players never to reveal any private information to the other
players (e.g., trade secrets).
In game theory, this concept has been noted before when discussing correlations.
In \cite{Forges93,Forges06}  it is called  ``conditional independence property'' or ``belief-invariance'',
while in \cite{Lehrer09} it is called ``non-communicating garbling'' and in \cite{Liu} is called ``individually uninformativeness''. In all cases above,
however, this property is used to make the analysis of the equilibria more convenient,
and is not highlighted as interesting in its own right.

From a completely different angle, belief-invariance has been a topic of research
in physics (motivated by questions in the foundations of quantum mechanics \cite{PR, Tsilerson, Barrett, masanes})
and theoretical computer science (motivated by multi-prover interactive
proof systems~\cite{KRR} and parallel repetition of games~\cite{buhrman2014parallel,FriedmanRV16,lancien2015parallel}),
under the name of \textit{non-signalling correlations}. In these investigations, belief-invariance is not used for making the analysis more convenient, but it is relevant in itself: it describes the largest class of correlations that obey relativistic causality.

\paragraph{Our contribution.}
In this work, we take the viewpoint currently adopted in physics and theoretical computer science literature, and we apply it to game theory. 
That is, we give a general picture of \emph{non-locality} as a resource:
this allows us to re-define the concepts of correlated, communication and \emph{belief-invariant equilibria} within a unified framework.
We remark that a similar approach has been recently taken in \cite{bergemann2016bayes}: as we show next, our framework turns out to be more ``flexible'', since it can be also adapted to include other classes of equilibria, such as \emph{quantum correlated equilibria}. The non-locality resource also suggests new reseach directions that may be of extreme interest for the community of algorithmic game theory: from the evaluation of complexity of the equilibria concepts discussed in this work, to the assessment of the performance of these equilibria, in a ``Price of Anarchy'' fashion. Interestingly, it turns out that the equilibria concepts discussed here may also be of interest and applications in other fields, as cryptography and privacy (see Appendix~\ref{sec:app}).

Specifically, in this work we first create the above mentioned framework: we give detailed definitions and rigorously prove some useful facts. The intention here is to provide a reference for future work, a unified review of the existing concepts with a focus on coherent and complete presentation that facilitates comparison. As mentioned above, the treatment of belief-invariance has a different viewpoint with respect to previous literature,
and this novel approach is essential in the rest of the paper.
In Appendix~\ref{apx:discuss} we also discuss here the complexity of equilibria defined in this framework and the relationship of them with cryptography and privacy.

Next, we put the framework to use by studying how the non-locality resources affect the performance of games.
In particular, we focus on \emph{social welfare} as a measure of such performances.
Our analysis highlights that belief-invariant equilibria can be socially better than both correlated equilibria and non-belief-invariant communication equilibria.
This sheds new lights on the concept of belief-invariant equilibria:
they are not just a technical tool for the computation of equilibria or a more private subclass of communication equilibria.
They are an equilibrium concept of social relevance
that call for a further and deeper investigation of its properties. 

Then, we extend our framework by considering a class of probability distributions based on \emph{quantum mechanical effects}. We include in the framework a subclass of the belief-invariant correlations that gives rise to the interesting concept of ``fully private correlations'':
correlations that reveal the type of a player neither to the other players nor to the mediator, and still can achieve equilibrium outcomes that seem to be impossible without an informed mediator.  This part of the paper shows the flexibility of our framework: it can accomodate all kinds of correlations, and in the future might be useful to study equilibria with other desired properties.
The extension is discussed separately, because we need to introduce mathematical tools from quantum mechanics. The use of such tools is essential for the complete and correct modelling of the quantum equilibria, which fit in our framework in a straightforward but at first counter-intuitive way, as discussed in Section \ref{sec:qframework}.
We also  continue our study of social welfare, by proving that there are games for which the quantum equilibria have better social welfare
than the ones that may be achieved without quantum effects.
The use of quantum mechanical correlations in game theory has been studied in many forms and flavours,
as can be seen from the survey \cite{Guo08}, the references therein, and the more recent works \cite{Brunner13,Pappa15,PhysRevA.96.042340}.
However, to the best of our knowledge, this work is the first to unify communication, belief-invariant, quantum and correlated equilibria.

Finally, we conclude by describing a few research directions that we believe may be of interest for the theoretical computer science and the game theory community.

As described above, in this work we bring together the strands of thought coming from the different backgrounds: game theory, physics, and theoretical computer science.
The interdisciplinarity leaves us with the problem of choosing the language in which to formulate our
results, as we might end up not reaching anyone from either side of the
discipline divide. We have chosen to use a language familiar to (algorithmic) game theorists,
but hopefully not too far from the one used in physics.
We decided to use the language of (algorithmic) game theory when discussing
relevant examples from physics \cite{CHSH, Greenberger90, Brunner13, Pappa15}.
We introduce some physical language only when we talk about quantum equilibria, in a stand-alone section.
However, even in this case, we give an introduction to some of the basic mathematical notions and notations of quantum mechanics,
trying to be accessible and consistent with the previous sections.
We hope this effort will make physics literature more accessible to game theorists and computer scientists and stimulate further research.

\paragraph{Organization of the paper.}
In Section \ref{sec:prel} we define the objects we are
working on: first games of incomplete information and then various classes of correlations.
In Section \ref{sec:equilibria} we introduce our framework. We redefine within this framework
the classes of communication, correlated, and belief-invariance equilibria.
In Section \ref{sec:examples}, we discuss how the different kinds of equilibria affect the social welfare,
and we show that in some cases belief-invariant equilibria reach a better social outcome.
No quantum physics is involved in the sections above; only in Section \ref{sec:quantum}
we broaden the discussion to include the quantum case, and we present it in
a way that makes it fit into the framework we have developed up to that point.
Finally, we will discuss some open problems in the conclusions section.
Even if our motivations  in this work are principally theoretical and conceptual, we are convinced that the concepts discussed in this paper
can also have practical relevance: indeed, in Appendix~\ref{sec:app},
we will discuss some possible applications.

\section{Preliminaries}
\label{sec:prel}
In this section we define the basic concepts we need to discuss our classes of
equilibria. First we will define games of incomplete information and their strategies.
Later, we will define the notion of correlation and the classes of
probability distributions we need.

\subsection{Games with incomplete information}
\label{sec:games_def}
In this section we briefly introduce our notation for some basic concepts in Game Theory.

A \emph{game with incomplete information} $G$ is defined by the following objects:
\begin{itemize}
\item A finite set of \emph{players} $N$, of size $n$, usually $N=[n]$;
\item A finite set of \emph{type profiles} $T := T_1 \times \cdots \times T_n$;
\item A finite set of \emph{action profiles} $A := A_1 \times \cdots \times A_n$;
\item A prior probability distribution $P(\tv)$ on the type profiles $\tv\in T$;
\item For each player $i \in N$, a \emph{payoff function} $v_i \colon T \times A \rightarrow \mathbb{R}$.
\end{itemize}

In a game of incomplete information, the behaviour of the players is modelled as follows. A \emph{strategy} $g_i$ for the player $i$ is a map
from the information that $i$ knows to an action $a_i \in A_i$.
In the absence of any correlation or external advice, players can apply pure strategies or mixed ones.
A \emph{pure strategy} for player $i$ is a map $g_i \colon T_i \rightarrow A_i$,
meaning that players select an action based only on their type.
A \emph{mixed strategy} for player $i$ is a probability distribution
over pure ones, i.e. the function $g_i \colon T_i \rightarrow A_i$
becomes a random variable. If we want to make its distribution explicit,
we introduce the independent local random variables $\lambda_i$ with probability $\Lambda_i(\lambda_i)$
and we set $g_i = g_{i,\lambda_i} = g_i(\cdot,\lambda_i)$.
This random function describes a conditional probability distribution on
$A_i$ given $T_i$, denoted by slight abuse of notation as $g_i(a_i \mid t_i)$.

The game goes as follows. The types $\tv=(t_1,\ldots,t_n)$ are sampled according to $P$.
Each player $i$ learns his type $t_i$\footnote{\dio{We assume w.l.o.g. that each type $t_i \in T_i$ has non-zero probability to be observed, i.e., there is $\tv_{-i}$ such that $P(t_i, \tv_{-i}) > 0$.}}, uses his strategy $g_i$ to select an action $a_i \in A_i$,
and is awarded according to his payoff function $v_i$
(which in general depends also on the other players' actions and types).
Hence, the expected utility of player $i$ \dio{who observed type $t_i$} is:
\begin{equation}\begin{split}
  \label{eq:payoff}
  \avg{v_{i\dio{,t_i}}(\gv)} &= \EE_{\dio{\tv_{-i} \mid t_i},\gv} v_i\bigl(\tv,g_1(t_1),\ldots, g_n(t_n)\bigr)
             = \sum_{\dio{\tv_{-i}},\av} P(\dio{\tv_{-i} \mid t_i}) v_i(\tv,\av) \prod_{i=1}^n g_i(a_i \mid t_i),
\end{split}\end{equation}
where $\gv = (g_1, \ldots, g_n)$ and $\av = (a_1, \ldots, a_n)$.
Observe that if $P$ is a point mass on a fixed type
$t_0$, $P(t_0)=1$, we recover the usual games of complete information.

A game is called \emph{full coordination game} if all the payoff functions are equal,
i.e., all players are interested in the same outcome.
On the other hand, we talk about
a \emph{game of conflicting interests} if there exists a type profile $\tv$, action profiles $\av\neq \bv$, and players $i,j$ so that $\av$ maximizes the utility of $i$ given $\tv$, and $\bv$ maximizes the utility of $j$ given $\tv$. In other words, the players can be interested in different outcomes.

A \emph{solution} for a game is a family of strategies $\gv = (g_1, \ldots, g_n)$, one for each player.
A solution is then said to be an \emph{equilibrium} (more precisely a
\emph{Nash equilibrium})\footnote{\dio{This is also known in literature as interim equilibrium, in order to distinguish it from ex-ante equilibrium, where players compute their expected utility before knowing their types.}} if no player has an incentive to change the adopted strategy.
In the basic, uncorrelated, case, this means that
\begin{equation}
  \avg{v_{i\dio{,t_i}}(\gv)} =    \EE_{\dio{\tv_{-i} \mid t_i},\gv} v_i\bigl( \tv,\gv_{-i}(\tv), g_i(t_i) \bigr)
            \geq \EE_{\dio{\tv_{-i} \mid t_i},(\gv_{-i},\chi_i)} v_i\bigl( \tv,\gv_{-i}(\tv_{-i}), \chi_i(t_i) \bigr) = \avg{v_{i\dio{,t_i}}(\gv_{-i},\chi_i)},
\end{equation}
for all $i \in N$, $t_i \in T_i$ and $\chi_i \in A_i^{T_i}$.
This can be expressed more concisely as saying that for all $i$, $t_i$ and $a_i$,
\begin{equation}
  \sum_{\tv_{-i},\lamv} P(\dio{\tv_{-i} \mid t_i}) \Lambda(\lamv) v_i(\tv,\gv_{-i}(\tv_{-i},\lamv_{-i})g_i(t_i,\lambda_i))
    \geq \sum_{\tv_{-i},\lamv} P(\dio{\tv_{-i} \mid t_i}) \Lambda(\lamv) v_i(\tv,\gv_{-i}(\tv_{-i},\lamv_{-i})a_i),
\end{equation}
where $\lamv = (\lambda_1, \ldots, \lambda_n)$, and, by independence of $\lambda_i$,
$\Lambda(\lamv) = \prod_{i=1}^n \Lambda_i(\lambda_i)$.

By Nash's theorem~\cite{Nash}, every game of incomplete information
has an equilibrium. In fact,
usually -- except in the simplest situations -- \dio{for each type profile $\tv$} we have to
expect several Nash equilibria $\gv$ to exist, with different payoff profiles
$(\avg{v_{i\dio{,t_i}}(\gv)}:i\in N)$.

\dio{The expected welfare $\SW_i(\gv)$ of player $i$ in a solution $\gv$ is
$$
 \SW_i(\gv) = \EE_{t_i}\avg{v_{i,t_i}(\gv)} := \avg{v_{i}(\gv)}.
$$}
The expected \emph{social welfare} $\SW(\gv)$ of a solution $\gv$
is the sum of the expected payoffs of all players,
$\SW(\gv) = \sum_i \dio{\SW_i(\gv)}$, and often used as a measure
of the quality of an equilibrium.
More generally, we may consider some other function $v(\tv,\av)$
of the types and actions
(e.g., the max of the expected payoffs of all players is the measure that is commonly adopted in the job scheduling setting \cite{ArcherT01}),
and look at the \emph{social payoff} $\SPO$ defined as follows:
\[
  \SPO(\gv) = \avg{v(\gv)} = \EE_{\tv,\gv} v\bigl( \tv,\gv(\tv) \bigr).
\]

\paragraph{Example: {\tt CHSH} game \cite{CHSH}.}
\label{sec:CHSH}
We now give an example of game of incomplete information, that we will use as our running example.
The game is a classic example from quantum information,
that later was also used in game theory (cf.~\cite{Forges06,Lehrer09,Liu}).
The game is called {\tt CHSH} after the authors of \cite{CHSH}.

It is a two-player game, with respective types and actions of players $(t_1, t_2)$ and $(a_1,a_2)$.
Both players' types and actions are single bits, taking values from the set $\{0,1\}$.
The distribution $P$ on the types $(t_1,t_2)$ is uniform, i.e. probability
$\frac{1}{4}$ is assigned to each of the four possibilities.
The game is a full coordination game, i.e., the payoff functions
$v_1,v_2$ are equal and the players want to achieve a common goal.
The payoffs are as in Figure \ref{fig:CHSH}.
\begin{figure}[htb]
\centering
\begin{game}{2}{2}[$t_1\cdot t_2=0$]
    & 0    & 1 \\
0   &1,1   &0,0\\
1   &0,0   &1,1
\end{game}
\qquad \qquad
\begin{game}{2}{2}[$t_1 \cdot t_2=1$]
    & 0    & 1 \\
0   &0,0    &1,1\\
1   &1,1    &0,0
\end{game}
\caption{The {\tt CHSH} game}
\label{fig:CHSH}
\end{figure}
In other words, the two players prefer to correlate if at least one
of them has the first type (which happens with probability $\frac{3}{4}$),
and to anti-correlate if they both have the second type (with probability $\frac{1}{4}$).

A simple pure strategy for each player in the {\tt CHSH} game is the constant
function mapping to $0$. Since the distribution $P$ is uniform, such joint action
by players gives an expected
\dio{utility of $1$ if observed type is 0, and $\frac{1}{2}$ otherwise}.
It is not hard to see that
this cannot be improved with other pure or mixed, or correlated strategies.
Later we will see that, in presence of external mediating devices, other equilibria exist
that reach the optimal expected utility of $1$\dio{, regardless of the observed type}.

\subsection{Correlations: joint conditional probability distributions}
Looking at eq.~(\ref{eq:payoff}), we see that for a well-defined expected
payoff, we only need a joint distribution of $\tv=(t_1,\ldots,t_n)$ and
$\av=(a_1,\ldots,a_n)$. In fact, since the marginal distribution $P(\tv)$ of
the types is fixed, we only require a conditional distribution of $\av$ given $\tv$.

This motivates us to consider, as a resource in gameplay,
a general \emph{correlation}, i.e.~a joint conditional probability distribution
\[
  Q(s_1, \ldots, s_n \mid r_1, \ldots, r_n)
\]
where $r_i$ and $s_i$ are defined as \emph{inputs} and \emph{outputs} for player $i$.
Here, we keep our discussion about correlations as general as possible, and we do not specify the meaning of inputs and outputs. Next section will provide more details for the setting of games of incomplete information.
We may abbreviate the notation as $Q(\sv \mid \rv)$
for tuples $\rv = (r_1, \ldots, r_n) \in R = \bigtimes_i R_i$
and $\sv = (s_1, \ldots, s_n) \in S = \bigtimes_i S_i$, where
$R_i$ and $S_i$ are the input and output alphabets of player
$i$, respectively.

Note that we do not assume any restriction on these correlations,
apart from the obvious requirements of probability distributions:
\[
  \sum_\sv Q(\sv\mid\rv) = 1 \quad \forall \rv.
\]
Given a set $R$ of inputs and a set $S$ of outputs we denote as $\mathtt{ALL}(S\mid R)$,
the set of all possible correlations on these sets.

By imposing additional restrictions we can single out
other meaningful subclasses of correlations,
that we will use later to define different kinds of equilibrium.

\paragraph{Belief-invariant (aka non-signalling) correlations.}
A joint conditional probability distribution $Q$ is \emph{belief-invariant}
(also called \emph{non-signalling})
if the distribution of the output variable $s_i$ given $r_i$ does not give any information about $r_j$, with $j\neq i$.
This class is easily seen to be strictly contained in the general class of correlations.
Indeed, the belief-invariant condition is clearly violated in a correlation where
$s_i$ could be just equal to $r_j$, i.e.,  $\Pr\{s_i = r_j\} = 1$.

The names belief-invariant and non-signalling can be understood in the following way.
Suppose we have $n$ parties, with the $i$-th party having access only to $r_i, s_i$.
Then the observation of $r_i, s_i$ does not reveal anything more about the other parties' $r_j$ variables than $r_i$ alone.
Therefore if the parties had a estimation (belief) of what could be the others' variables,
this is not changed by the observation of the outputs of the correlation $Q$.

Formally, for a set $I \subset N$, let $R_I = \bigtimes_{i\in I} R_i$ and $S_I = \bigtimes_{i\in I} S_i$.
Then, we say that a correlation $Q(\sv \mid \rv)$ is belief-invariant \cite{KRR} for all
subsets $I \subset N$ and $J = N\setminus I$,
\begin{equation} \label{eq:NS}
 \sum_{\sv_J \in S_J} Q(\sv_I,\sv_J\mid \rv_I,\rv_J) = \sum_{\sv_J \in S_J}  Q(\sv_I,\sv_J \mid \rv_I,\rv'_J) \
 \forall \sv_I \in S_I, \ \rv_I \in R_I, \ \rv_J,\rv'_J \in R_J.
\end{equation}
Given a set $R$ of inputs and a set $S$ of outputs we denote as $\mathtt{BINV}(S \mid R)$,
the set of all belief-invariant correlations on these sets.

We remark that this class of correlations has various equivalent definitions (see, for example, \cite{masanes}),
but we prefer the one given above since it makes it clear that any subset of parties,
even when getting together, cannot learn anything more about the other subset's input
variables than what they would know from the joint distribution of their $\{R_i\}_{i\in I}$ alone.

\paragraph{Local correlations.}
A joint conditional probability distribution $Q$ is called \textit{local}
if it can be simulated locally by each party $i$, by observing (their part of)
a random variable $\gamv=(\gamma_1, \ldots, \gamma_n)$ (with distribution $V(\gamv)$)
independent of $\rv$, and doing local operations depending only on $r_i$ and $\gamma_i$.
More formally, a correlation $Q(\sv \mid \rv)$ is local if there exists a random variable $\gamv$
and distributions $L_i(s_i \mid r_i \gamma_i)$ such that
\begin{equation}
  \label{eq:LOC}
  Q(\sv \mid \rv) = \sum_{\gamv} V(\gamv) L_1(s_1 \mid r_1 \gamma_1) \cdots L_n(s_n \mid r_n \gamma_n).
\end{equation}
Any local distribution is also belief-invariant, because the condition \eqref{eq:NS} is respected.
However, the opposite is not true, meaning that the inclusion is strict.
An example of non-local belief-invariant distribution is given below in \eqref{eq:chsh_ns_strategy}.

\medskip
As above, given a set $R$ of inputs and a set $S$ of outputs we denote as $\mathtt{LOC}(S\mid R)$,
the set of all local correlations on these sets.
We remark that the sets $\mathtt{ALL}(S\mid R)$, $\mathtt{BINV}(S\mid R)$ and $\mathtt{LOC}(S\mid R)$ are all closed convex sets.

\section{Equilibria with communication and correlation resources}
\label{sec:equilibria}
Consider a game $G=(N,T,A,P,\{v_i\})$ as defined in Section \ref{sec:games_def}.
A solution with communication for $G$ studies the behaviour of players
who have access to a \emph{correlation device} that depends on inputs
communicated by them during the game.
The most common operational interpretation of this setting is that
a \emph{trusted mediator}, who has private communication channels with all the players,
collects from each player $i$ the input $r_i$, samples $\sv$ according to $Q(\sv\mid\rv)$
and sends to each $i$ the output $s_i$.

Formally, we add to the strategies of the players the use of a \textit{correlation} $Q(\sv \mid \rv)$,
where $\rv = (r_1, \ldots, r_n) \in R$ is a tuple of \textit{inputs} and
$\sv = (s_1, \ldots, s_n) \in S$ is a tuple of \textit{outputs}.
In this setting, a \textit{pure strategy} for each player $i$ is a pair of functions,
$f_i \colon T_i \rightarrow R_i$ and
$g_i \colon T_i \times S_i \rightarrow A_i$; and a mixed strategy is a pair of jointly distributed random functions
$(f_i,g_i) \in R_i^{T_i} \times A_i^{T_i \times S_i}$. As done above, the latter can be given explicitly by specifying
a local random variable $\lambda_i$ for each player $i$
with distribution $\Lambda_i(\lambda_i)$, and letting $f_i = f_{i,\lambda_i} = f_i(\cdot,\lambda_i)$,
$g_i = g_{i,\lambda_i} = g_i(\cdot,\lambda_i)$.
The joint distribution of $\lamv = (\lambda_1,\ldots,\lambda_n)$ is a product distribution,
$\Lambda(\lamv) = \prod_{i=1}^n \Lambda_i(\lambda_i)$,
reflecting the fact that the $n$ pairs $\{(f_i,g_i)\colon i=1,\ldots,n\}$ are independent.

The game now goes as follows.
The types $\tv = (t_1, \ldots, t_n)$ are sampled according to $P$.
Each player~$i$ learns his type $t_i$,
and sends the input $r_i = f_i(t_i)$ to the correlation device.
He then gets the correlation output $s_i$ and plays the action $a_i = g_i(t_i, s_i)$.
This makes all of $\tv$, $\rv$, $\fv$, $\gv$ and $\av$ jointly distributed
random variables.
The expected payoff of player $i$ \dio{who observed type $t_i$} is:
\[\begin{split}
  & \avg{v_{i\dio{,t_i}}(\gv)} = \EE_{\dio{\tv_{-i} \mid t_i},\fv,\gv,Q} v_i\bigl(\tv,g_1(t_1,s_1),\ldots,g_n(t_n,s_n)\bigr) \\
            & \quad = \sum_{\dio{\tv_{-i}},\sv,\lamv} P(\dio{\tv_{-i} \mid t_i}) \Lambda(\lamv) Q\bigl(\sv\mid f_1(t_1,\lambda_1),\ldots,f_n(t_n,\lambda_n)\bigr)
                                            v_i\bigl(\tv,g_1(t_1,s_1,\lambda_1),\ldots,g_n(t_n,s_n,\lambda_n)\bigr).
\end{split}\]

We now give the definitions of some classes of equilibria, introduced by Forges \cite{Forges82},
that are meaningful in this setting with communication:
communication equilibria, belief-invariant communication equilibria and correlated equilibria.
A new class, namely quantum equilibria, will be defined later in Section \ref{sec:quantum}.
Communication and correlated equilibria were explicitly discussed in previous
work, notably \cite{aumann1974subjectivity,Forges82,Forges93,Forges06}.
The intermediate class of belief-invariant (and the quantum variant of Section \ref{sec:quantum})
was previously discussed only indirectly in some works, for example \cite{Forges93,Lehrer09,Pappa15}.

\subsection{Communication equilibrium}
\label{sec:bayesiancomm}
The most general class we consider here is the class of \emph{communication equilibria}.
Here, we assume that the correlation device can implement a correlation $Q$ that is \emph{unrestricted}.
We will obtain later two meaningful subclasses by restricting the class of available correlations.

In order to formally define a solution for a game we need to describe
not only the correlation $Q$ implemented by the correlation device,
but also the strategies, i.e., the functions $\{f_i\}$ and $\{g_i\}$,
and the private randomness used by the players.
To this aim, given an $n$-tuple $\xv = (x_1, \ldots, x_n)$, we use the
standard abbreviation $\xv_{-i}$ to denote the $(n-1)$-tuple
$(x_1, \ldots, x_{i-1}, x_{i+1}, \ldots,x_n)$, i.e., $\xv$ with the
$i$-th entry removed.
Similarly, if $\fv = (f_1,\ldots,f_n)$ is a family of functions in which each
$f_i$ is a function of an argument $x_i$, we denote with
$\fv_{-i} = (f_1,\ldots,f_{i-1},f_{i+1},\ldots,f_n)$ the family with the
$i$-th member removed, and by $\fv_{-i}(\xv_{-i})$ the tuple of values
$f_{1}(x_{1}),\ldots, f_{i-1}(x_{i-1}),f_{i+1}(x_{i+1}),\ldots, f_n(x_n)$.

\begin{definition}[Communication equilibrium]
\label{def:communication}
A solution $(\fv,\gv,Q)$ is a communication equilibrium of a game
$G$ if for each player $i$, \dio{each type $t_i \in T_i$,} and for all random functions $\varphi_i \colon T_i \rightarrow R_i$ and $\chi_i \colon T_i \times S_i \rightarrow A_i$,
\begin{equation}
\label{eq:incentive}
\begin{split}
  \sum_{\dio{\tv_{-i}},\lamv,\sv}
    & P(\dio{\tv_{-i}\mid t_i}) \Lambda(\lamv) Q(\sv \mid f_i(t_i,\lambda_{i})\fv_{-i}(\tv_{-i},\lamv_{-i}))
         v_i(t,g_i(t_i,s_i,\lambda_{i})\gv_{-i}(\tv_{-i},\sv_{-i},\lamv_{-i}))  \\
    &\geq \sum_{\dio{\tv_{-i}},\lamv,\sv} P(\dio{\tv_i \mid t_i}) \Lambda(\lamv) Q(\sv \mid \varphi_i(t_i,\lambda_{i})\fv_{-i}(\tv_{-i},\lamv_{-i}))
                              v_i(t,\chi_i(t_i, s_i,\lambda_{i})\gv_{-i}(\tv_{-i},\sv_{-i},\lamv_{-i})) .
\end{split}
\end{equation}
\end{definition}

This definition captures the idea of having no incentive to deviate
unilaterally, but it may seem a formidable task to verify the conditions
it imposes.
Note that it is w.l.o.g. to assume that in the above definition
$\varphi_i$ and $\chi_i$ are deterministic, i.e., not depending on $\lambda_i$.
Hence, it easily follows that we do not have to
go over all functions $\varphi_i$, but it is sufficient to go over all its possible outputs $r_i$.
Hence, we have the following alternative definition.

\begin{definition}
  \label{def:communication-simpler}
  A solution $(\fv,\gv,Q)$ is a communication equilibrium of a game
  $G$ if for all $i$, $t_i$, $r_i$ and functions $\chi_i \colon T_i \times S_i \rightarrow A_i $,
  \begin{align*}
    & \sum_{\tv_{-i},\lamv,\sv} P(\dio{\tv_{-i} \mid t_i}) \Lambda(\lamv) Q(\sv\mid\fv(\tv,\lamv) v_i(\tv,\gv(\tv,\sv,\lamv))\\
    & \qquad \geq \sum_{\tv_{-i},\lamv,\sv} P(\dio{\tv_{-i} \mid t_i}) \Lambda(\lamv)
                              Q(\sv \mid r_i \fv_{-i}(\tv_{-i},\lamv_{-i}))
                              v_i(\tv, \chi_i(t_i,s_i) \gv_{-i}(\tv_{-i},\sv_{-i},\lamv_{-i})).
  \end{align*}
\end{definition}

\paragraph{The canonical form.}
Definition \ref{def:communication} is useful for understanding the subclasses of equilibria we define later.
However, it is possible to express the communication equilibria in a simpler
\emph{canonical form},
where players communicate their types (not a function of the type) to the correlation device,
and the latter returns the actions they have to take (not only an information from which players compute their action).
The intuition is that the mediator, who implements the correlation
$Q(\sv\mid\rv)$ also takes care of the computation of the functions $r_i=f_i(t_i)$
and $a_i=g_i(t_i,s_i)$.
Starting from a general communication solution $(\fv,\gv,Q)$,
we construct a new \emph{canonical solution} $(\id^{(f)},\id^{(g)},\widehat{Q})$, where $\id_i^{(f)}(t_i) = t_i$, $\id_i^{(g)}(t_i, s_i) = s_i$ for each $i \in N$, $t_i \in T_i$ and $s_i \in S_i$, and the correlation $\widehat{Q}(\av \mid \tv)$ works as follows:
\begin{equation}
  \label{eq:canonical}
  \widehat{Q}(\av\mid\tv)
    = \sum_\lamv \Lambda(\lamv) \sum_{\sv \colon \gv(\tv,\sv,\lamv)=\av} Q(\sv \mid \fv(\tv,\lamv)).
\end{equation}
This is often called the \emph{revelation principle}. 
In what follows, we will simplify the notation by setting $(\id^{(f)},\id^{(g)},\widehat{Q})$ as $(\id,\id,\widehat{Q})$.
It is clear that the expected payoffs for this solution
when the players truthfully reveal their type and take the suggested action
(that is, when their strategy corresponds to identity functions $\id$ both for inputs and outputs)
are the same as those of the original solution\dio{, i.e., for each $i$ and each $t_i \in T_i$}
\[
  \avg{v_{i\dio{,t_i}}(\gv)} = \sum_{\dio{\tv_{-i}},\av} P(\dio{\tv_{-i} \mid t_i}) \widehat{Q}(\av\mid\tv) v_i(\tv,\av).
\]
Furthermore, we have the following important proposition.
\begin{proposition}
  \label{prop:canonical-comm-eq}
  If $(\fv,\gv,Q)$ is a communication equilibrium, then its associated canonical
  solution $(\id,\id,\widehat{Q})$ is also a communication equilibrium with exactly the same outcome. In words, no player has an incentive to communicate a false
  type, or to take an action different from the one suggested.
\end{proposition}
\begin{proof}
Since $(\fv,\gv,Q)$ is an equilibrium, no player has an incentive in communicating a false type or in taking an action different from the one that has been suggested.
If in $\widehat{Q}$ there is a player $i$ who can increase his expected payoff by deviating from the suggested action on a type $t_i$, then the same deviation would increase the expected payoff of player $i$ in $(\fv,\gv,Q)$ for $t_i$. This contradicts the assumption that $(\fv,\gv,Q)$ is an equilibrium.
It follows that $(\id,\id,\widehat{Q})$ is an equilibrium as well. Because it preserves the conditional distribution of actions given types of $(\fv,\gv,Q)$, it also preserves the outcome.
\end{proof}
Notice that there are (infinitely) many equilibria $(\fv,\gv,Q)$ that lead to
the same canonical solution $(\id, \id, \widehat{Q})$. In fact, \eqref{eq:canonical} and
Proposition~\ref{prop:canonical-comm-eq} imply an equivalence
relation on solutions. Also notice that each equivalence class of a communication
equilibrium contains exactly one canonical solution, which we call
the \emph{canonical representative} of the class.

Note also that, since the communicated information comprises only types and actions,
these are the only two deviations that a player can take.
Therefore, the above discussion implies that we can simplify the notion of communication equilibrium as follows.

\begin{definition}[Canonical communication equilibrium]
  \label{def:communication_canonical}
  A solution $(\id, \id, Q)$ is a canonical communication equilibrium
  if there exists an equilibrium $(\fv,\gv,Q')$ in its equivalence class,
  i.e.~if the former is the canonical solution associated to the latter.

  Equivalently, $(\id, \id, Q)$ is a
  canonical communication equilibrium
  if for all $i \in N$, $t_i, r_i \in T_i$ and $\chi_i \colon T_i \times A_i \rightarrow A_i$,
  \begin{equation}
    \label{eq:incentive2}
    \sum_{\tv_{-i},\av_{-i}} P(\dio{\tv_{-i} \mid t_i}) Q(\av\mid\tv) v_i(\tv,\av)
       \geq \sum_{\tv_{-i},\av_{-i}} P(\dio{\tv_{-i} \mid t_i}) Q(\av \mid (r_i, \tv_{-i})) v_i(\tv,\chi_i(t_i, a_i)\av_{-i}).
  \end{equation}
\end{definition}

The above observations imply that as far as communication equilibria are
concerned, we may without loss of generality restrict our attention to their
canonical representatives. Then, for given game $G$, we denote as $\COMMEQ(G)$
the set of canonical communication equilibria for $G$.
The next proposition shows that, for every game, this is a convex set.
\begin{proposition}
  \label{prop:comm-equilibria-convex}
  If $(\id, \id, Q_1)$ and $(\id, \id, Q_2)$ are canonical communication equilibria for the same
  game $G$, then so is $(\id, \id, Q)$, with $Q = p Q_1 + (1-p) Q_2$, for $0\leq p\leq 1$.
\end{proposition}
\begin{proof}
Consider the following correlation, with $\av = (a_1,\ldots,a_n) \in A$,
$\bv = (b_1,\ldots,b_n) \in \{0,1\}^n$ and $\tv = (t_1,\ldots,t_n) \in T$:
\[
  Q'(\av,\bv\mid\tv) = \begin{cases}
                         p Q_1(\av\mid\tv)    & \text{ if } b_1=b_2=\ldots=b_n=0, \\
                         (1-p)Q_2(\av\mid\tv) & \text{ if } b_1=b_2=\ldots=b_n=1, \\
                         0                 & \text{ otherwise}.
                       \end{cases}
\]
That is, $Q'$ chooses with probability $p$ and $1-p$ to
provide $Q_1$ and $Q_2$, respectively, and informs the players
of its choice along with the recommended actions $\av$.
Thus it is clear from the fact that $(\id, \id, Q_1)$ and $(\id, \id, Q_2)$ are
equilibria, that $(\id, \id, Q')$ (where players do not output the $b_i$ variables) is an equilibrium, too.
The proof is concluded by observing that $Q = \widehat{Q'}$ is
the canonical representative of $Q'$.
\end{proof}

\paragraph{An example.}
Consider again the {\tt CHSH} game described above.
A communication equilibrium for the {\tt CHSH} game consists in players
revealing their type to the mediator, receiving information about which game are they playing
(i.e., information about the type of the other player),
and which action they are suggested to play in that game,
and following the advice of the mediator.
Formally, consider inputs $(r_1, r_2) \in T$ and outputs $(s_1, s_2)$ where $s_i \in A_i \times \{0,1\}$.
Moreover, consider a correlation $Q$ such that $Q(00, 00 \mid r_1 \cdot r_2 = 0) = 1$ and $Q(01, 11 \mid r_1 \cdot r_2 = 1) = 1$.
Then, the solution $(\id, s_1, Q)$, where $s_1$ is the function that returns the first bit of the advice $s$ received by the player,
is a communication equilibrium, since all players have expected payoff $1$ \dio{regardless of the observed type,} and no action can be taken in order to increase this payoff.

Note also that this equilibrium is not canonical, since the mediator is not just suggesting an action, but more complex advices are given.
Anyway, it is immediate to transform above equilibrium in a canonical one,
by requiring the mediator to return only the action players are suggested to take.

\subsection{Belief-invariant equilibrium}
We obtain the subclass of belief-invariant equilibria
by requiring that the correlation used in the equilibrium is in the
class of belief-invariant correlations.

\begin{definition}[Belief-invariant equilibrium]
\label{def:beliefinvariant}
A solution $(\fv,\gv,Q)$ is called \emph{belief-invariant} if
$Q$ is a belief-invariant correlation. If $(\fv,\gv,Q)$ is a communication
equilibrium, we call it a \emph{belief-invariant (communication) equilibrium}.

A solution $(\id, \id, Q)$ is a \emph{canonical belief-invariant equilibrium}
if there is a belief-invariant equilibrium $(\fv,\gv,Q')$ in its equivalence class, i.e.,
with $Q=\widehat{Q'}$.
\end{definition}

The relation between belief-invariant equilibria and their canonical
version is clarified in the following proposition, whose proof is evident and, hence, omitted.

\begin{proposition}
  \label{prop:canonical-eq-beliefinv}
  If $(\fv,\gv,Q)$ is a belief-invariant equilibrium and $\widehat{Q}$
  its canonical representative, then $(\id,\id,\widehat{Q})$ is also
  a belief-invariant equilibrium, with the same \dio{outcome} as the
  original equilibrium.
\end{proposition}

In a caveat to the above proposition, we stress that the equivalence class of a
canonical belief-invariant equilibrium $(\id, \id, Q)$
may contain also solutions involving non-belief-invariant correlations.
Indeed, the correlating device could provide players with information about other players
even if this information is useless with respect to the choice of the strategy to play.
Below we discuss an example showing how this can occur.

Finally, observe that the canonical belief-invariant equilibria of a given
game $G$, denoted $\BIEQ(G)$, form
a convex set, like the canonical communication equilibria:

\begin{proposition}
  \label{prop:beliefinv-equilibria-convex}
  If $(\id, \id, Q_1)$ and $(\id, \id, Q_2)$ are canonical belief-invariant equilibria for the same
  game, then so is $(\id, \id, Q)$, with $Q = p Q_1 + (1-p) Q_2$ for $0\leq p\leq 1$.
\end{proposition}
\begin{proof}
Just notice that $Q'$ as defined in the proof of
Proposition~\ref{prop:comm-equilibria-convex} is belief-invariant.
\end{proof}

\paragraph{An example.}
Consider again the {\tt CHSH} game described above.
Recall the canonical communication equilibrium $Q$ previously described,
i.e. $Q(00\mid t_1 \cdot t_2 = 0) = 1$ and $Q(01\mid t_1 \cdot t_2 = 1) = 1$.
It is easy to see that this equilibrium is not belief-invariant.
Indeed, whenever the second player receives advice $1$,
his belief about the type of the first player changes,
since he knows for sure that it is $1$.

Still, there is a canonical belief-invariant equilibrium.
Consider, indeed, the following correlation $Q'$:
$Q'(00\mid t_1 \cdot t_2 = 0) = Q'(11\mid t_1 \cdot t_2 = 0) = 1/2$
and $Q'(01\mid t_1 \cdot t_2 = 1) = Q'(10\mid t_1 \cdot t_2 = 1) = 1/2$.
Note that each player receives a payoff of $1$\dio{, regardless of the observed type,} and, thus, there is no way for them to improve their utility.
Anyway, each player receives advice $a$ with probability $1/2$ regardless of the other player's type.
Hence, this correlation does not allow players to gain any information about other players.

We note that there are multiple equilibria that belong to the equivalence class whose canonical representative is $Q'$.
In some of these equilibria the advice received by players can contain information that does not serve to the player
for computing the action, but still reveal information about the other player's type.
Consider for example the following correlation:
$Q(00, 00 \mid r_1 \cdot r_2 = 0) =
Q(10, 10 \mid r_1 \cdot r_2 = 0) = 1/2$ and
$Q(01, 11 \mid r_1 \cdot r_2 = 1) =
Q(11, 01 \mid r_1 \cdot r_2 = 1) = 1/2$,
in which the advice not only suggests the action that the player should take,
but also reveal which game the players are actually playing,
and thus, which is the type of the other player.

\subsection{Correlated equilibrium}
We obtain the subclass of correlated equilibria by
requiring that the correlation used at the equilibrium is essentially
a shared random variable.

\begin{definition}[Correlated equilibrium]
\label{def:correlated}
A solution $(\fv,\gv,Q)$ is called \emph{correlated}
if the output distribution of $Q$ is independent of the input:
$Q(\sv\mid\rv) = Q(\sv)$ for all $\rv$ and $\sv$. If it is a communication
equilibrium, we speak of a \emph{correlated (communication) equilibrium}.

Similarly to the belief-invariant case, we transfer the property
of being correlated to the canonical representative $\widehat{Q}$,
and we call a canonical correlation $Q$ a \emph{canonical correlated equilibrium}
if there is any correlated equilibrium $(\fv,\gv,Q')$ in its equivalence class,
i.e., $Q=\widehat{Q'}$.
\end{definition}

The set of canonical correlated equilibria of the game $G$
is denoted $\CORREQ(G)$.

Our definition is equivalent to the one of Forges (see \cite{Forges82}
and \cite[page 8]{Forges93}),
who describes the correlated equilibrium as a collection
$\bigl( Q(\sv), g_1, \ldots, g_n \bigr)$,
where $Q$ is a distribution of suggestions independent of the types,
and each $g_i \colon T_i\times S_i \rightarrow A_i$ is a function that
player $i$ uses to determine their action. Indeed, note that the
functions $f_i$ in a correlated solution serve no purpose, since
the input $r_i = f_i(t_i)$ to $Q$ is irrelevant, and only the sampled
output $s_i$ and its correlation with $\sv_{-i}$ matter. Hence,
from now on we will denote a correlated solution/equilibrium simply
as $(\gv,Q)$, with a probability distribution $Q$ on $S$.

This also allows us to exhibit a simplified equilibrium criterion.
\begin{proposition}
  \label{prop:equilibrium-corr-simplified}
  A tuple $(\gv,Q)$ is a correlated equilibrium if and only if
  for all $i \in N$, $t_i \in T_i$, $s_i \in S_i$, and $a_i \in A_i$,
  \[
    \sum_{\tv_{-i},\sv_{-i}} P(\dio{\tv_{-i} \mid t_i}) Q(\sv) v_i(\tv,\gv(\tv,\sv))
       \geq \sum_{\tv_{-i},\sv_{-i}} P(\dio{\tv_{-i} \mid t_i}) Q(\sv) v_i(\tv,a_i\gv_{-i}(\tv_{-i},\sv_{-i})).
  \]
\end{proposition}

\begin{remark}
\normalfont
If $(\gv,Q)$ is a correlated equilibrium, then
its canonical correlation $\widehat{Q}$ is local in the sense of
\eqref{eq:LOC}. Be aware, however, that we do not know if
$(\id,\id,\widehat{Q})$ is a correlated equilibrium,
because it may not have the required property that $\widehat{Q}(\av\mid \tv)$ is
independent of $\tv$ (in contrast to
Proposition~\ref{prop:canonical-eq-beliefinv} in the belief-invariant case).

Worse, even if $(\id,\id,Q)$ is a communication equilibrium such
that $Q$ is a local correlation, it is not clear whether this implies
that $\widehat{Q}$ is a canonical correlated equilibrium in the sense
of Definition~\ref{def:correlated}.
To show this, one would have to find a correlated equilibrium $(\gv,Q')$
such that $Q=\widehat{Q'}$ is its canonical representative.
The difficulty stems from the fact that while $Q$ can be simulated by
giving a suitable shared randomness $\Gamma = (\Gamma_1,\ldots,\Gamma_n)$
to cooperating players, see eq.~(\ref{eq:LOC}),
to competing players it might give an advantage over the others
having access to $\Gamma$ directly rather than only $Q$.
\end{remark}

To have a usable handle on correlated equilibria, we propose
the following definition, which allows us to identify the
correlated solutions to a game of incomplete information with the correlated
strategies of the associated game of complete information.

\begin{definition}
\label{def:correlated-standardform}
We say that a correlated solution (equilibrium) $(\gv,Q)$
is \emph{in standard form}
if for all $i$, the advice space equals the set of functions
from types to actions, $S_i = A_i^{T_i}$, and if the function
$g_i$ consists of evaluating the first argument (a function) on
the second argument:
\begin{align*}
  g_i : T_i \times A_i^{T_i} &\longrightarrow A_i, \\
              (t_i,\sigma_i) &\longmapsto \sigma_i(t_i).
\end{align*}
Clearly, such a solution is given entirely by the distribution
$Q$ on $A_1^{T_1} \times \cdots \times A_n^{T_n}$, which we will
thus use as a shorthand for a correlated solution in standard form.
\end{definition}
The following proposition shows that the correlated equilibria in games with
incomplete information, or more specifically their standard form,
are precisely the correlated equilibria in the associated game
of complete information that has the strategy space $S_i=A_i^{T_i}$
for player $i$. The proof is evident and, hence, omitted.

\begin{proposition}
  \label{prop:standardform-eq-correlated}
  If $(\gv,Q)$ is a correlated equilibrium, then we can obtain a
  correlated equilibrium $\widetilde{Q}$ in standard form that is
  in the same equivalence class, as follows:
  \[
    \widetilde{Q}(\sigv) := \Pr_Q\bigl\{ \forall i\ g_i(\cdot,s_i)=\sigma_i \bigr\}.
  \]
  Thus, a canonical equilibrium $Q$ is a correlated equilibrium if and
  only if there exists a correlated equilibrium in standard form
  in its equivalence class.
  \qed
\end{proposition}

\begin{remark}
\normalfont
Note that we can also define Nash equilibria in this formalism.
Nash equilibria are precisely the correlated equilibria $(\gv,Q)$
with a product distribution $Q(\sv) = Q_1(s_1)\cdots Q_n(s_n)$.

Note that it is straightforward to see that if $(\gv,Q)$ is a Nash equilibrium,
then the canonical representative is also a product distribution, i.e.,
$\widehat{Q} = \widehat{Q}_1 \times \cdots \times \widehat{Q}_n$, and
in fact a Nash equilibrium.
Conversely, if $Q(\av\mid\tv)$
is a canonical communication equilibrium that factorizes,
i.e., $Q(\av\mid\tv) = Q_1(a_1\mid t_1) \cdots Q_n(a_n\mid t_n)$, then there is a
Nash equilibrium in its equivalence class with
the same payoffs. This is obtained by writing each of the local
transition probabilities $Q_i(a_i\mid t_i)$ as probabilistic mixtures
of functions in $A_i^{T_i}$.

Thus, in the spirit of previous definitions,
we can speak of a \emph{canonical Nash equilibrium} as a factorizing
canonical communication equilibrium $Q(\av\mid\tv)$.
We will denote the set of canonical Nash equilibria of a game $G$ as $\NEQ(G)$.
\end{remark}

As expected, and as it should be, the set of correlated equilibria is convex:
\begin{proposition}
  \label{prop:corr-equilibria-convex}
  If $(\gv^{(1)},Q_1)$ and $(\gv^{(2)},Q_2)$ are correlated equilibria
  for the same game, with $Q_j$ a distribution on $S^{(j)}$
  and $0\leq p \leq 1$, then so is $(\gv,Q)$ with
  \begin{align*}
    Q(\sv,\bv)   &:= \begin{cases}
                       p Q_1(\sv)     & \text{ if } b_1=\ldots=b_n=0,\ \sv \in S^{(1)}, \\
                       (1-p) Q_2(\sv) & \text{ if } b_1=\ldots=b_n=1,\ \sv \in S^{(2)}, \\
                       0              & \text{ otherwise};
                     \end{cases}    \\
    g_i(s_i,b_i) &:= g^{(b_i+1)}(s_i).
  \end{align*}
\end{proposition}
\begin{proof}
The proof follows from the insight that $(\gv,Q)$ gives each player the information which
of the two solutions $(\gv^{(1)},Q_1)$ and $(\gv^{(2)},Q_2)$
was implemented, so any benefit from deviating from the advice
for player $i$ would imply an advantage for the player
in one of $(\gv^{(1)},Q_1)$ or $(\gv^{(2)},Q_2)$, contradicting
the assumption that they are equilibria.
\end{proof}

\paragraph{An example.}
Consider again the {\tt CHSH} game defined above.
It is easy to check that the belief-invariant equilibrium described above is not correlated.
Indeed, a shared random variable is not sufficient to understand which game players are actually playing
and to understand if they need to coordinate or to anti-coordinate.

An example of correlated equilibrium for this game is $(\id, Q)$ with $Q(00)= Q(11) = 1/2$.
Here, the action of each player does not depend on the other player's type.
Anyway, it is still convenient for the player to follow the suggestion, since they will receive payoff $1$ \dio{with probability $1$ if their observed type is 0, and
with probability $1/2$ otherwise,} and there is no alternative action that allows them to receive this payoff with larger probability.
Observe that this equilibrium is not a Nash equilibrium, since it requires shared randomness and cannot be factorized.
A Nash equilibrium with the same payoffs is $(\id, Q)$ with $Q(00)= 1$.

\subsection{General properties of the equilibrium classes} \label{sec:properties}
The main goal of this work is to highlight the relations between the
different concepts of equilibria we have introduced so far.
In this section we will start this analysis by first discussing a general inclusion between the equilibrium classes, and then by giving examples of games for which these classes coincide.
Later, in Section~\ref{sec:examples}, we will continue our analysis by comparing the performance of
these equilibria with respect to social welfare maximization.
Further properties and applications can be found in the Appendix.

\medskip

\paragraph{Inclusions} We have observed that there are correlated equilibria that are not Nash,
belief-invariant equilibria that are not correlated,
and communication equilibria that are not belief-invariant.
Since communication, belief-invariant and correlated
equilibria have convexity properties, we can arrange the canonical versions of these equilibria
into nested sets within the set of canonical correlations $Q(\av\mid\tv)$:
\begin{equation}
  \label{eq:inclusions}
  \NEQ(G) \subset \operatorname{conv}\bigl(\NEQ(G)\bigr)
                \subset \CORREQ(G)
                \subset \BIEQ(G)
                \subset \COMMEQ(G),
\end{equation}
where $\operatorname{conv}\bigl(\NEQ(G)\bigr)$ denotes the convex hull of the set of Nash equilibria for $G$.

Actually, this inclusion can be easily derived from the similar structure existing between the
different classes of correlation resources behind these equilibrium concepts.

Indeed, by our previous observations, $\CORREQ(G)$ is always a convex
subset of the local correlations $\mathtt{LOC}(A\mid T)$, $\BIEQ(G)$
is a convex subset of the non-signalling correlations $\mathtt{BINV}(A\mid T)$,
and all are contained in the set of all correlations $\mathtt{ALL}(A\mid T)$.
The inclusion structure between these classes is known (see, e.g., \cite{masanes}) for full-coordination games\footnote{In quantum physics and computer science, these games are known as non-local games.
While they may seem uninteresting because of the lack of competition, they are a wonderful way of reasoning about the classes of correlations.
The objective of a non-local game, in our parlance here, is to find
an optimal equilibrium $\gv$ with respect to a certain payoff function
$v(\tv,\av)$ common to all players, which simply boils down to
optimizing $\dio{\frac{\SW(\gv)}{|N|}} = \sum_{\tv,\av} P(\tv)Q(\av\mid \tv) v(\tv,\av)$
over all canonical solutions $Q$ from a given class.}.

\paragraph{Situations in which communication or correlation are useless.}
If our game $G$ is really one of complete information in the sense that
the types of the players are deterministically prescribed, i.e.,
\dio{$T_i = \{t^*_i\}$ for each player $i$ and thus $P(\tv^*) = 1$},
then any communication equilibrium
is equivalent to a correlated equilibrium.
More precisely, for the canonical form $Q$ of the communication
equilibrium $\gv$, $Q^* := Q(\cdot\mid\tv^*)$ as a probability distribution on
$A = A_1 \times \cdots \times A_n$ is a correlated equilibrium,
\dio{such that each player has an expected payoff}
\begin{align*}
  \avg{v_{i\dio{,t_i^*}}(\gv)} & = \sum_{\dio{\tv_{-i}},\av} P(\dio{\tv_{-i} \mid t^*_i}) Q(\av\mid(\tv_{-i},t_i^*)) v_i((\tv_{-i},t_i^*),\av)\\
            & = \sum_{\av} Q(\av\mid\tv^*) v_i(\tv^*,\av)
            = \sum_{\av} Q^*(\av) v_i(\tv^*,\av).
\end{align*}
This reproduces the result
of Zhang~\cite{Zhang12}, that quantum (and indeed \emph{any})
correlation doesn't change the landscape of equilibria in games
of complete information beyond correlated equilibria.

\medskip
{Secondly, we can identify a simple class of games where the presence of
no additional correlation changes the set of equilibrium payoffs;
we call them (two-player) \emph{\dio{symmetric} constant sum games \dio{with uniform prior}}. They are
characterized by the property that \dio{$|T_1| = |T_2| = \Theta$ and for every $t_1, t_2, a_1, a_2$}
\[
  v_1(t_1 t_2, a_1 a_2) + v_2(t_1 t_2, a_1 a_2) = s(t_1 t_2),
\]
i.e., the sum of the two players' individual payoffs depends
only on the type. \dio{Moreover, this game priors are uniform, i.e., $P(t_1, t_2) = \frac{1}{\Theta^2}$ for every $t_1, t_2$. Note that this implies that $P(t_1 \mid t_2) = P(t_2 \mid t_1) = \frac{1}{\Theta}$ for every $t_1, t_2$.}
This implies that for every $t_1, t_2$, regardless of the solution $\gv$
employed,
\begin{align*}
  \avg{v_{\dio{1,t_1}}(\gv)} + \avg{v_{\dio{2,t_2}}(\gv)} & = \dio{\sum_{t'_2 \in T_2,\av} P(t'_2 \mid t_1) Q(\av \mid (t_1,t'_2)) v_1((t_1,t'_2),\av)}\\
  &\dio{\qquad \qquad + \sum_{t'_1 \in T_1,\av} P(t'_1 \mid t_2) Q(\av \mid (t'_1,t_2)) v_2((t'_1,t_2),\av)}\\
  & = \dio{\frac{1}{\Theta} \sum_{\tv,\av} Q(\av \mid \tv) [v_1(\tv,\av) + v_2(\tv,\av)]}\\
  & = \dio{\frac{1}{\Theta} \sum_{\tv} s(\tv)} =: \avg{s}.\footnotemark
\end{align*}
\footnotetext{\dio{Note that $\avg{s} \neq \EE_{\tv}[s]$ (specifically, $\avg{s} = \Theta\EE_{\tv}[s]$). However, we slight abuse notation for sake of readability.}}
Thus, the game is an instance of a zero-sum game according to the
theory of von Neumann and Morgenstern~\cite{vNeumann-Morgenstern}
(we can make the sum of payoffs explicitly equal to zero by subtracting
$s(t_1t_2)$ from $v_1(t_1t_2,a_1a_2)$, but we refrain from doing so
not to overload notation).

This means that there exists a (mixed) strategy $Q^{\rm vNM}_i(a_i\mid t_i)$
for player $i$ \dio{who observed type $t_i$}, which guarantees him a payoff $v^{\rm vNM}_{i,\dio{t_i}}$ regardless
of what the other player does.
\dio{Let $\avg{v^{\rm vNM}_{i,t_i}} = \EE_{\tv_{-i} \mid t_i,(Q^{\rm vNM}_1,Q^{\rm vNM}_2)}[v_i(t_i,\tv_{-i},a_1,a_2)] = \Theta v^{\rm vNM}_{i,t_i}$. Then we have that}
$\dio{\avg{v^{\rm vNM}_{1,t_1}}}+\dio{\avg{v^{\rm vNM}_{2,t_2}}} = \avg{s}$.
In particular, if we consider any communication equilibrium $\gv$ of the
game, then the first player's payoff cannot increase if the second player
were to use his von-Neumann-Morgenstern strategy (ignoring the advice)
that guarantees him a payoff of at least $v^{\rm vNM}_{2,\dio{t_2}}$. Thus, for every $t_1, t_2$
\[
  \avg{v_{1,\dio{t_1}}(\gv)} \leq \avg{s} - \dio{\avg{v^{\rm vNM}_{2,t_2}}} = \dio{\avg{v^{\rm vNM}_{1,t_1}}},
\]
and symmetrically
\[
  \avg{v_{2\dio{,t_2}}(\gv)} \leq \avg{s} - \dio{\avg{v^{\rm vNM}_{1,t_1}}} = \dio{\avg{v^{\rm vNM}_{2,t_2}}}.
\]
This generalizes a result of~\cite{Brunner13}: in the section titled ``a game where none of the payoff functions is a 
Bell inequality'' the authors discuss a game without quantum advantage which is \dio{symmetric} constant-sum \dio{and has uniform prior}.}

\section{Impact of correlation on social welfare}
\label{sec:examples}
In this section we show that no-signalling correlation can have a positive impact on the social welfare of a game.
Specifically, there are games in which a belief-invariant equilibrium can achieve a social welfare that is better than every correlated equilibria.
The {\tt CHSH} game discussed above gives us a clear example of this fact:
indeed, we showed there is a belief-invariant equilibrium that achieves an expected social welfare of $2$,
whereas it is not hard to see that no correlated equilibrium is better that the one described above,
whose social welfare is $3/2$.

However, the {\tt CHSH} game is a two-player full coordination game, and one can wonder whether such a result
holds even if we consider games with conflict of interests and/or with more than two player.
Pappa et al.~\cite{Pappa15} give a partial answer to this question, by showing that
the above result holds for a two-player conflict-of-interest variant of the {\tt CHSH} game.
Below, we report this result for completeness.
Moreover, we extend their result by presenting a $n$-player game with conflict of interests in which
a belief-invariant equilibrium exists that is better than any correlated equilibrium.
Interestingly, our game is a variant of the {\tt GHZ} game,
a game motivated from quantum mechanics \cite{Greenberger90}.

Since the class of belief-invariant equilibria strictly contains the class of correlated equilibria,
it may be expected that the former contains equilibria that are better than the ones in the latter class.
It is instead surprising that a correlated equilibrium
can perform better
than any other communication equilibrium.
However, we next show that this may be the case.
In other words, we prove that locality
is
not only
a desirable requirement,
but
it is
sometimes necessary in order to achieve high social welfare. Note that a general form of Pappa et al.\ game has been studied in \cite{roy2016quantum}, where the
authors have shown a quantum advantage in the context of social welfare.

\subsection{Belief-invariant equilibria can outperform correlated equilibria}
\subsubsection{Two-player games with conflict of interests}
A modified version of {\tt CHSH} has been used in \cite{Pappa15} to obtain a
two-player game with conflict of interests in which there is a belief-invariant equilibrium
that achieves a better expected social welfare than any correlated equilibrium.
We report it here for completeness.

In this game, the players are still interested in coordinating or anti-coordinating as in {\tt CHSH},
but now each player prefers a specific outcome, as follows:
\renewcommand{\gamestretch}{1.5}
\begin{figure}[htb]
\centering
\begin{game}{2}{2}[$t_1\cdot t_2=0$]
    & 0    & 1 \\
0   &1,$\frac{1}{2}$   &0,0\\
1   &0,0   &$\frac{1}{2}$,1
\end{game}
\qquad \qquad
\begin{game}{2}{2}[$t_1 \cdot t_2=1$]
    & 0    & 1 \\
0   &0,0    &$\frac{3}{4}$,$\frac{3}{4}$\\
1   &$\frac{3}{4}$,$\frac{3}{4}$    &0,0
\end{game}
\caption{The game of Pappa et al. \cite{Pappa15}}
\label{fig:CHSH2}
\end{figure}
\renewcommand{\gamestretch}{1}

%

The pure strategies $(0,0)$ and $(1,1)$ lead to two equilibria with unfair expected payoffs, in a battle-of-sexes flavour.\footnote{Battle of Sexes is a classic game theory example. It is used in many textbooks, for example \cite{fudenberg1991game,NisaRougTardVazi07}.}
No player has incentive to deviate from constant actions $(0,0)$, but the fist player has 
\dio{expected payoff $1$ if the observed type is $0$, and $\frac{1}{2}$ otherwise, while the other player has expected payoff $\frac{1}{2}$ if the observed type is $0$, and $\frac{1}{4}$ otherwise.}
For the second equilibrium, constant $(1,1)$, we have the same unfairness, this time in favor of the second player.

The situation can be improved with the notions of communication equilibria we discussed in Section \ref{sec:equilibria}.
With a correlated equilibrium we have a solution similar to battle-of-sexes: with one bit shared randomness one can select either the first or the second pure equilibrium uniformly. This makes the situation fair, with an expected payoff for each player of \dio{$\frac{3}{4}$ if the observed type is $0$, and $\frac{3}{8}$ otherwise (hence the expected welafare of each player is $\frac{9}{16}$, and the expected social welfare is $\frac{9}{8}$)}.
There is also an unfair correlated equilibria where the two players get \dio{expected welfare} of $\frac{11}{16}$ and $\frac{7}{16}$, respectively.
However, the following belief-invariant correlation guarantees to both players a fair and optimal expected \dio{welfare} of $\frac{3}{4}$:
\begin{equation}
\label{eq:chsh_ns_strategy}
\begin{split}
\text{ If $ t_1\cdot t_2=0$  then }  \quad Q(0,0\mid t) = Q(1,1\mid t) = \frac{1}{2}, \\
\text{ if $ t_1\cdot t_2=1$  then } \quad  Q(0,1\mid t) = Q(1,0\mid t) = \frac{1}{2}.
\end{split}
\end{equation}
This is belief-invariant because the marginal of each player is a uniformly random bit whatever the other player's type is. Notice that this correlation solves perfectly the common objective of {\tt CHSH}, i.e., coordinating if $ t_1\cdot t_2=0$ and anti-coordinating otherwise. Also, in the case $ t_1\cdot t_2=0$ it behaves like a correlated equilibrium in the battle-of-sexes game, by selecting one of the two pure strategies $(1,1)$ or $(0,0)$ uniformly.

As said above, this correlation is used in \cite[page 335]{Forges06}, and it is well-known in the physics community as the PR-box \cite{PR}.
It is also known that this belief-invariant correlation cannot be implemented as a local one \cite{Tsilerson}.

\subsubsection{$n$-player games with conflict of interests}
\label{sec:ghz}
We now introduce a game based on an example in physics known as \emph{the GHZ state} for three parties \cite{Greenberger90}. (The result can be generalized to $n$ parties, we chose $n=3$ for simplicity.)

We have three players, each one with two possible types (which we label type 0 and 1) and two possible actions (action 0 and 1). The possible type triples $(t_1, t_2,t_3)$ are taken from the set $\{ (0,0,1), (0,1,0), (1,0,0), (1,1,1) \} $ with probability $$p(0,0,1) = p(0,1,0) = p(1,0,0) = \frac{1}{6}, \qquad p(1,1,1) = \frac{1}{2}.$$
Let $\tau = t_1 \cdot t_2 \cdot t_3$. We have $\Pr(\tau=1) = \Pr(\tau=0) = \frac{1}{2}$.
The payoff are given \dio{in Figure~\ref{tab:GHZ}}.
\begin{figure}[htb]
\centering
\begin{minipage}{0.5\textwidth}
\begin{game}{2}{2}[0]
    & 0    & 1 \\
0   &$\eps$,$\eps$,$\eps$   &0,0,0\\
1   &0,0,0   &$\eps$,1,1
\end{game}
\qquad
\begin{game}{2}{2}[1]
    & 0    & 1 \\
0   &0,0,0    &1,$\eps$,1\\
1   &1,1,$\eps$    &0,0,0
\end{game}
\subcaption{$\tau = 0$}
\end{minipage}
\begin{minipage}{0.5\textwidth}
\begin{game}{2}{2}[0]
    & 0    & 1 \\
0   &0,0,0   &1,1,$\eps$\\
1   &1,$\eps$,1   &0,0,0
\end{game}
\qquad
\begin{game}{2}{2}[1]
    & 0    & 1 \\
0   &$\eps$,1,1    &0,0,0\\
1   &0,0,0    &$\eps$,$\eps$,$\eps$
\end{game}
\subcaption{$\tau = 1$}
\end{minipage}
\caption{A modified {\tt GHZ} game. The subgame (a) is played when $\tau = 0$ and subgame (b) when $\tau = 1$.
In both cases the strategy player 1 identify the table (0 for right table, 1 for the left table),
the strategy of player 2 chooses the row (0 for the top row, 1 for the bottom row),
and the strategy of player 3 chooses the column (0 for left column, 1 for right column).
Within a cell, the first value is the payoff of player 1, the second value is the payoff of player 2 and the last payoff is for player 3.}
\label{tab:GHZ}
\end{figure}

Thus, the players jointly lose the game (have all payoff 0) whenever $\tau \neq a_1 + a_2 + a_3  \mod 2$.  In the non-losing cases, the players whose action is equal to $\tau$ receive payoff $\eps$ (a positive number very close to 0), while the others receive payoff 1.

Therefore this game, in the spirit of \cite{Pappa15}, features both coordination and conflicting interests. The players are jointly interested in minimizing the probability of having payoff 0, while each player individually dislikes to be the one implementing the action $\tau$ in the winning cases.

%

In the best correlated equilibrium (in terms of expected \dio{social welfare}),
the mediator suggests to each player $i$ the function $$\sigma_i = 1-t_i \mod 2.$$
This always wins in the case $\tau = 0$ and loses in the case $\tau = 1$.
It gives \dio{expected social welfare $\frac{2+\eps}{2}$}.

The best communication equilibrium is as follows.
The mediator learns the types, and if $\tau = 0$ he suggests actions $(0,1,1), (1,0,1), (1,1,0)$ uniformly at random,
while if $\tau = 1$ he suggests actions $(0,0,1), (0,1,0), (1,0,0)$ uniformly at random.
This gives \dio{expected social welfare $2+\eps$}.


This communication equilibrium is not belief-invariant,
because the marginals for the players' actions are not the same
in case $\tau = 0$ and $\tau = 1$. For example, if a player has type 1
and receives advice for action 0, then his belief will change,
assigning more probability to the case $\tau = 1$.

Thus, in any belief-invariant equilibria,
the mediator must make sure that, for all the possible triples of types,
the marginal distributions of all players are the same.
Since the payoff of player $i$ is maximized when the action $a_i \neq \tau$,
the expected social welfare for a belief-invariant equilibrium is maximized by considering a distribution $Q$
such that the $i$-th marginal gives $\Pr_{Q, \tv}(a_i = 1 \mid t_i) = \Pr_{Q, \tv}(a_i = 0 \mid t_i) = \frac{1}{2}$, whatever the type $t_i$ is.
This is implemented by the following distribution $Q$:
\begin{equation}
\label{eq:ghz_ns_strategy}
\begin{split}
& \text{ If $\tv \in \{ (0,0,1), (0,1,0), (1,0,0) \}$  then }  \\ & \qquad \qquad  Q(0,0,0\mid \tv) = Q(0,1,1\mid \tv) =Q(1,0,1\mid \tv) =Q(1,1,0\mid \tv) = \frac{1}{4}, \\
& \text{ if $\tv = (1,1,1) $  then } \\ & \qquad \qquad  Q(0,0,1\mid \tv) =Q(0,1,0\mid \tv) =Q(1,0,0\mid \tv) =Q(1,1,1\mid \tv) = \frac{1}{4}.
\end{split}
\end{equation}
It is easy to see that this is a belief-invariant equilibrium,
since any deviating player would decrease his own expected payoff by deviating
(it makes everyone lose in at least a value of $\tau$).

One can check that in this equilibrium \dio{the expected social welfare is $\frac{3}{4} (2 + \varepsilon)$}
that is better than the expected \dio{social welfare} of the best correlated equilibrium.
There are biased communication equilibria produced through an unrestricted, non-private, correlation. Such an equilibrium can have \dio{expected social welfare} as large as \dio{3}.

\subsection{Belief-invariant equilibria can outperform non-belief-invariant ones}

Consider the following two-player game of incomplete information:
the two players with types $t_1, t_2 \in \{0,1\}$.
Each player has also two available actions, also named $0$ and $1$.
The payoffs are as follows:
\begin{figure}[htb]
\centering
\begin{game}{2}{2}[$t_1\cdot t_2=0$]
    & 0    & 1 \\
0   &1-$\eps$,1-$\eps$   &2,0\\
1   &0,2   &2-$\eps$,2-$\eps$
\end{game}
\qquad \qquad
\begin{game}{2}{2}[$t_1 \cdot t_2=1$]
    & 0    & 1 \\
0   &0,0    &1,1\\
1   &1,1    &0,0
\end{game}
\caption{A game in which belief-invariant equilibria outperform non-belief-invariant ones}
\label{fig:PDLG}
\end{figure}

Thus, if $t_1\cdot t_2=0$, then payoffs resemble the ones of the Prisoners' Dilemma,\footnote{The Prisoners' Dilemma is another classic example in game theory. It is found in many textbooks, for example \cite{fudenberg1991game,NisaRougTardVazi07}.}
so that it is a dominant strategy for each player to take action $0$.
If instead $t_1\cdot t_2=1$ then players are playing a full coordination game in which they
prefer to take different actions.
We assume that each type profile $(t_1, t_2)$ has the same probability $\frac{1}{4}$ of being generated.

Let us consider a distribution $Q$ of the form $Q(\av \mid \tv)$.
We will show that for any such distribution, if $(\id, \id, Q)$ is an equilibrium
and maximizes the social welfare among all the equilibria,
then it is a correlated equilibrium, and hence it is belief-invariant.
Note that there is no loss of generality in considering only canonical equilibria.
Indeed, as stated above, if a non-canonical communication equilibrium $(\fv,\gv,Q')$ exists with a better social welfare,
then its canonical representative $(\id, \id, \widehat{Q'})$ is still an equilibrium
and has \dio{the same outcome and thus} the same social welfare as $(\fv,\gv,Q')$.

We start by stating conditions for $(\id,\id,Q)$ being an equilibrium that maximizes the social welfare.
A first simple observation is the following:
If $Q(1,a_2 \mid 0,t_2) > 0$, then $(\id,\id,Q)$ is not in equilibrium.
Indeed, when $t_1 = 0$, then player 1 knows with probability $1$ that $t_1\cdot t_2=0$,
and thus it is a dominant strategy to take action $0$.
By symmetry, the same observation holds by inverting the roles of players.
Hence, in order to have a canonical equilibrium we must have that $Q$ is as follows:
\begin{equation}
\label{eq:Q_def}
\begin{aligned}
 Q(0,0 \mid 0,0) = 1;\; Q(0,0 \mid 0,1) &= 1-p;\; Q(0,1 \mid 0,1) = p;\\
 Q(0,0 \mid 1,0) = 1-q;\;& Q(1,0 \mid 1,0) = q;\\
 Q(0,0 \mid 1,1) = p_{00}; \; Q(0,1 \mid 1,1) = p_{01}; &\; Q(1,0 \mid 1,1) = p_{10}; \; Q(1,1 \mid 1,1) = p_{11}.
 \end{aligned}
\end{equation}

Next lemma states conditions on these values in order for $Q$ being an equilibrium.
\begin{lemma}
\label{lem:ceq_cond}
$(\id, \id, Q)$ is an equilibrium if and only if the following conditions are satisfied:
\begin{align}
 p_{10} - p_{11} & \geq (1-\eps)q \label{ceq1}\\
 p_{01} - p_{11} & \geq (1-\eps)p \label{ceq2}\\
 p_{00} - p_{01} & \leq (1-\eps)(1-q) \label{ceq3}\\
 p_{00} - p_{10} & \leq (1-\eps)(1-p) \label{ceq4}\\
 (p_{01}+p_{10}) - (p_{00}+p_{11}) & \geq (1-\eps)(2q-1) \label{ceq5}\\
 (p_{01}+p_{10}) - (p_{00}+p_{11}) & \geq (1-\eps)(2p-1). \label{ceq6}
\end{align}
\end{lemma}
\begin{proof}
  By Definition~\ref{def:communication-simpler}, $(\id, \id, Q)$ is an equilibrium
 if and only if for every $i \in \{1,2\}$, $t_i \in \{0,1\}$, $b_i \in \{0,1\}$
 and any function $\chi_i \in A_i^{T_i \times A_i}$
 \begin{equation}
 \label{eq:ceq}
 \begin{aligned}
  \avg{v_{i\dio{,t_i}}(\id, \id, Q)} & := \sum_{\tv_{-i},\av} P(\dio{\tv_{-i} \mid t_i}) Q(\av\mid \tv_{-i}t_i) v_i(\tv,\av_{-i}a_i)\\
  & \geq \sum_{\tv_{-i},\av} P(\dio{\tv_{-i} \mid t_i}) Q(\av\mid \tv_{-i}b_i) v_i(\tv,\av_{-i}\chi_i(t_i,a_i)).
 \end{aligned}
 \end{equation}
 As observed above, when $t_i = 0$, \eqref{eq:ceq} easily holds since, whichever the advice is,
 it is a dominant strategy for player $i$ to take action $0$.
 Hence, we only need to verify that \eqref{eq:ceq} holds when $t_i = 1$.
 To this aim, observe that there are only four possible functions $\chi_i(1,\cdot)$:
 the identity function that sets $\chi_i(1,a_i) = a_i$,
 the two constant functions that set $\chi_i(1,a_i) = 0$ and $\chi_i(1,a_i) = 1$ for every $a_i$, respectively,
 and the negation function that sets $\chi_i(1,a_i) = 1 - a_i$.

 Thus, if $i=1$ and $b_i = t_i = 1$, then we require that
 \begin{align*}
  \avg{v_{1\dio{,1}}(\id, \id, Q)} & \geq \sum_{t_2,\av} P(\dio{t_2 \mid 1}) Q(\av\mid 1,t_2) v_i(\tv,0a_2) = \frac{1-\eps + p_{01} + p_{11}}{2}\\
  \avg{v_{1\dio{,1}}(\id, \id, Q)} & \geq \sum_{t_2,\av} P(\dio{t_2 \mid 1}) Q(\av\mid 1,t_2) v_i(\tv,1a_2) = \frac{p_{10} + p_{00}}{2}\\
  \avg{v_{1\dio{,1}}(\id, \id, Q)} & \geq \sum_{t_2,\av} P(\dio{t_2 \mid 1}) Q(\av\mid 1,t_2) v_i(\tv,(1-a_1)a_2) = \frac{(1-\eps)q + p_{00} + p_{11}}{2}.
 \end{align*}
 Similarly, if $i=1$ and $0 = b_i \neq t_i = 1$, then we require that
 \begin{align*}
  \avg{v_{1\dio{,1}}(\id, \id, Q)} & \geq \sum_{t_2,\av} P(\dio{t_2 \mid 1}) Q(\av\mid 0,t_2) v_i(\tv,a_1a_2) = \frac{1-\eps + p_{01} + p_{11}}{2}\\
  \avg{v_{1\dio{,1}}(\id, \id, Q)} & \geq \sum_{t_2,\av} P(\dio{t_2 \mid 1}) Q(\av\mid 0,t_2) v_i(\tv,0a_2) = \frac{1-\eps + p_{01} + p_{11}}{2}\\
  \avg{v_{1\dio{,1}}(\id, \id, Q)} & \geq \sum_{t_2,\av} P(\dio{t_2 \mid 1}) Q(\av\mid 0,t_2) v_i(\tv,1a_2) = \frac{p_{10} + p_{00}}{2}\\
  \avg{v_{1\dio{,1}}(\id, \id, Q)} & \geq \sum_{t_2,\av} P(\dio{t_2 \mid 1}) Q(\av\mid 0,t_2) v_i(\tv,(1-a_1)a_2) = \frac{p_{10} + p_{00}}{2}.
 \end{align*}
 Since $\avg{v_{1\dio{,1}}(\id, \id, Q)} = \frac{(1-\eps)(1-q) + p_{01} + p_{10}}{2}$,
 it follows that all these inequalities hold if and only if conditions \eqref{ceq1},\eqref{ceq3},\eqref{ceq5} are satisfied.

 By repeating the same argument, we can observe that all the inequalities regarding to the second player hold
 if and only if conditions \eqref{ceq2},\eqref{ceq4},\eqref{ceq6} are satisfied.
\end{proof}

Let us now consider the correlation $Q^\star$ that sets $p = q = p_{01} = p_{10} = 1/2$ and $p_{00} = p_{11} = 0$.
It is immediate to check that conditions (\ref{ceq1}-\ref{ceq6}) are satisfied,
and, hence, $(\id, \id, Q^\star)$ is an equilibrium.

Note also that, if $t_i = 0$, then the player $i$ is suggested to take action $0$, regardless of the other player's type,
whereas, if $t_i = 1$, the player $i$ uses a shared random variable to decide which action they have to take.
Thus, the solution $(\id, \id, Q^\star)$ is equivalent to $(\gv, \widetilde{Q^\star})$,
where $\widetilde{Q^\star}(\av \mid \tv) = \widetilde{Q^\star}(\av)$ for each $\av$, and, in particular,
sets $\widetilde{Q^\star}(00) = \widetilde{Q^\star}(11) = 0$ and $\widetilde{Q^\star}(01) = \widetilde{Q^\star}(10) = 1/2$,
and $\g = (g_1, g_2)$, with $g_i \in A_i^{T_i \times A_i}$ that sets $g_i(0,a_i) = 0$ and $g_i(1,a_i) = a_i$ for each $a_i \in A_i$.
In other words, $(\id, \id, Q^\star)$ is a correlated equilibrium.

Finally, it is easy to check that the expected social welfare of the correlated equilibrium $(\id, \id, Q^\star)$ is $2-\eps$.

We next state the main result of this section,
namely that for this game any communication equilibrium achieving an expected social welfare that is at least $2-\eps$
must be correlated.
\begin{theorem}
 Any canonical communication equilibrium $(\id, \id, Q)$ with expected social welfare at least $2-\eps$ is correlated.
\end{theorem}
\begin{proof}
 Let $Q$ be as in \eqref{eq:Q_def}. As observed above, any canonical communication equilibrium $(\id, \id, Q)$
 must be distributed as $Q$ and must satisfy conditions (\ref{ceq1}-\ref{ceq6}).

 The expected social welfare of a solution $(\id, \id, Q)$ is
 \begin{align*}
  \SW(\id,\id, Q) & = \frac{1}{4} \left[2(1-\eps) + (2(1-\eps)(1-p) + 2p) + (2(1-\eps)(1-q)+2q) + 2(p_{01} + p_{10})\right]\\
  & = \frac{1}{2} \left[3 - \eps[3 - p - q] + (p_{01} + p_{10})\right].
 \end{align*}
 Thus, the expected social welfare of $(\id, \id, Q)$ is at least $2-\eps$ only if
 $$
 p_{00} + p_{11} = 1 - p_{01} - p_{10} \leq (p + q - 1) \eps.
 $$
 However, since $p_{00} + p_{11} \geq 0$, we have that
 \begin{equation}
 \label{eq:p+q_ub}
  p+q \geq 1.
 \end{equation}
 Moreover, since $p \leq 1$, we also have that
 $$
  p_{11} \leq p_{00} + p_{11} \leq (p + q - 1) \eps \leq q\eps,
 $$
 and, similarly, $p_{11} \leq p\eps$.

 Then, for a canonical communication equilibrium $(\id, \id, Q)$ with social welfare at least $2-\eps$,
 we have, from \eqref{ceq1}, that $p_{10} \geq (1-\eps)q + p_{11} \geq q$,
 and, from \eqref{ceq2}, that $p_{01} \geq p$.
 Thus, $p+q \leq p_{01} + p_{10} \leq 1$.
 Hence and from \eqref{eq:p+q_ub} we can then conclude that $q = 1-p$, $p_{10} = q$ and $p_{01} = p$.

 Therefore in any communication equilibrium $(\id, \id, Q)$ that achieves welfare at least $2-\eps$,
 the distribution $Q$ is as follows:
 \begin{align*}
  Q((0,0) \mid (0,0)) & = 1;\\
  Q((0,0) \mid (0,1)) = Q((1,0) \mid (1,0)) = Q((1,0) \mid (1,1)) & = 1-p;\\
  Q((0,1) \mid (0,1)) = Q((0,0) \mid (1,0)) = Q((0,1) \mid (1,1)) & = p;
 \end{align*}
 whereas all the remaining probabilities are $0$.

That is, if $t_i = 0$, then the player $i$ is suggested to take action $0$,
and, if $t_i = 1$, then the player $i$ uses a shared random binary variable with distribution $(p, 1-p)$ to decide which action they have to take.
Thus, as above, we can rewrite the equilibrium $(\id, \id, Q)$ as $(\gv, \widetilde{Q})$,
where $\widetilde{Q}(\av \mid \tv) = \widetilde{Q}(\av)$ for each $\av$, and, in particular,
sets $\widetilde{Q}(00) = \widetilde{Q}(11) = 0$, $\widetilde{Q^\star}(01) = p$ and $\widetilde{Q^\star}(10) = 1-p$,
and $\g = (g_1, g_2)$, with $g_i \in A_i^{T_i \times A_i}$ that sets $g_i(0,a_i) = 0$ and $g_i(1,a_i) = a_i$ for each $a_i \in A_i$.
In conclusion, $(\id, \id, Q^\star)$ is a correlated equilibrium.
\end{proof}

We remark that a communication equilibrium that is not correlated exists in this game.
Consider, indeed, the correlation $Q'$ that sets $p = q = 0$, $p_{01} = p_{10} = 1/2$ and $p_{00} = p_{11} = 0$.
We have that $(\id, \id, Q')$ is an equilibrium since it satisfies (\ref{ceq1}-\ref{ceq6}).
However, if the type of player 1 is $1$, then she is suggested to take action $0$
with a larger probability when $t_2 = 0$ than when $t_2 = 1$.
Thus, $(\id, \id, Q')$ is not belief-invariant, and, as a consequence it is not correlated.
However, as proved above, $(\id, \id, Q')$ has an expected social welfare of $2 - \frac{3}{2} \eps$
that is lower than the social welfare of $(\id, \id, Q^\star)$.

\section{Quantum equilibria}
\label{sec:quantum}
In this section we define the class of quantum correlated solutions and
equilibria of games, which is a physically motivated subclass of belief-invariant
solutions and equilibria.
This class is interesting because it allows us to go beyond the local correlations
without the need of an informed trusted mediator. Rather, like the correlated
equilibria, it only requires the (still trusted) distribution of a prior shared
resource: a quantum state. To make use of it, the players need the capability
to store this quantum information, and to make measurements on it at will,
according to their type. The measurement result then informs their action.

\subsection{Quantum formalism}
In the interest of being self-contained, we briefly present the mathematical
formalism necessary to introduce and discuss quantum correlated equilibria.
(For more details, see \cite{nielsen2000quantum, wilde2013quantum}.)
This section is essential for the correct understanding of the definitions that follow, and it contains standard concepts of quantum mechanics presented in a concise way for the reader not acquainted with the topic. We use standard notation, therefore the reader who is already at comfort with quantum formalism can safely skip to Section \ref{sec:qeq}.

Mathematically, a \emph{quantum state} is given by a \emph{density operator}
$\rho$ acting on a complex Hilbert space $\mathcal H$,\footnote{In the finite case with $d$ dimensions one can take $H=\mathbb C^d$.}
which means that $\rho$ is positive semidefinite, $\rho \geq 0$,
and has unit trace, $\Tr(\rho) = 1$. Every density operator can
occur as the state of a system.

A \emph{measurement} on the quantum system with Hilbert space $\mathcal{H}$
is given by a \emph{resolution of the identity},
or \emph{positive operator valued measure (POVM)}, which is a collection $(M_s : s\in S)$
of positive semidefinite matrices $M_s\geq 0$, one for each possible
outcome $s\in S$ of the measurement, acting on $\mathcal{H}$ and such
that $\sum_s M_s = \1$. Every resolution of the identity can be
realized in a quantum mechanical experiment.

States and measurements are the way quantum theory encodes the observable
features of physical systems. The fundamental formula is \emph{Born's rule},
which determines the probability of observing an outcome:
\begin{equation}
  \Pr\{ s \mid \rho \} = \tr \rho M_s.
\end{equation}
The above-mentioned rules for states and measurements make sure that
these numbers are always nonnegative, and add up to $1$ for each state
and each measurement. This formalism includes classical probability
theory, by restricting to diagonal states $\rho$ and measurement
operators $M_s$ (in some fixed, ``computational'' orthonormal basis $\{\ket{x}\}$):
\begin{align*}
  \rho &= \sum_x r_x \proj{x},    \\
  M_s  &= \sum_x \mu_s(x) \proj{x},
\end{align*}
where $\proj{x}$ denotes the projector onto the line $\CC\ket{x}$,
according to the famous Dirac notation of row vectors and column vectors \cite{dirac_1939}, aka ``bra-ket notation''. In this case,
the conditions for a state are equivalent to $r_x \geq 0$ and $\sum_x r_x = 1$,
i.e., $(r_x)_x$ is a probability vector; the conditions for a measurement
reduce to $\sum_s \mu_s(x) = 1$ for all $x$, i.e., $[\mu_s(x)]_{s,x}$
is a stochastic matrix.

Note that for simplicity, we assume the discrete setting here: the Hilbert spaces
are all finite dimensional and the measurements have discrete sets of outcomes.

In our multi-player games, we associate a Hilbert space $\mathcal{H}_i$ to
each player's quantum system, while their joint quantum system is
described by the tensor product Hilbert space
$\mathcal{H} = \mathcal{H}_1 \ox \cdots \ox \mathcal{H}_n$.
If each player has a measurement $M^i = (M^i_{s_i}:s_i \in S_i)$ acting on
$\mathcal{H}_i$, we associate to them a joint measurement
$M^1 \ox \cdots \ox M^n = \bigl(M_{\sv}:\sv = (s_1,\ldots,s_n) \in S\bigr)$
acting on $\mathcal{H}$:
\[
  M_{(s_1,\ldots,s_n)} = M^1_{s_1} \ox \cdots \ox M^n_{s_n}.
\]
Then, for a state $\rho$ acting on the composite Hilbert space $\mathcal{H}$,
\[
  \Pr\{ \sv = (s_1,\ldots,s_n) \mid \rho \} = \tr \rho(M^1_{s_1} \ox \cdots \ox M^n_{s_n}).
\]

To make the link with the correlations discussed in the previous sections,
consider the situation that every player $i$ has access to several measurements
$M^{r_i}$, $r_i \in R_i$, for simplicity all with outcomes in a set $S_i$,
which however could be specific to the player.
Now, given a state $\rho$ and the measurements $M^{r_i}$, the probability of
outputs $s_1,\ldots,s_n$ given the players' inputs $r_1, \ldots, r_n$ is
\begin{equation}
  \label{eq:qcorr}
  Q(s_1,\ldots,s_n \mid r_1, \ldots, r_n) = \Tr \rho (M^{r_1}_{s_1} \otimes \cdots \otimes M^{r_n}_{s_n}).
\end{equation}
This is evidently a correlation, and the correlations that can
be written in the above form, with a suitable state and suitable measurements,
are called \emph{quantum correlations}, their set denoted $\text{Q}(S\mid R)$.
(See, for example, the definition in the survey \cite{palazuelos2015random}.)

\begin{fact}[Belief-invariance of quantum correlations]
  The correlation $Q$ obtained as in eq.~\eqref{eq:qcorr}
  is always belief-invariant.
\end{fact}
\begin{proof}
Let $I$ and $J = N\setminus I$ be a partition of $N$.
Recall that for all $j$ and $r_j \in R_j$ we have  $\sum_{s_j} M^{r_j}_{s_j} = \1$.
Thus, for all $s_I \in S_I, \ r_I \in R_I, \ r_J,r'_J \in R_J$,
\begin{align*}
 \sum_{s_J \in S_J} q(s_I,s_J\mid r_I,r_J)
 & =  \sum_{s_J \in S_J}  \Tr \rho \left( \bigotimes_{i\in I} M^{r_i}_{s_i}
                                          \ox \bigotimes_{j\in J} M^{r_j}_{s_j} \right) \\
 & =  \Tr \rho \left(\bigotimes_{i\in I} M^{r_i}_{s_i} \ox \bigotimes_{j\in J} \1 \right) \\
 & =  \sum_{s_J \in S_J}  \Tr \rho \left( \bigotimes_{i\in I} M^{r_i}_{s_i}
                                          \ox \bigotimes_{j\in J} M^{r'_j}_{s_j} \right) \\
 & =  \sum_{s_J \in S_J} q(s_I,s_J\mid r_I,r_J'),
\end{align*}
and we are done.
\end{proof}

\begin{remark}
  \label{rem:classical-quantum}
  \normalfont
  Any local correlation $Q$,
  \[
    Q(\sv\mid\rv) = \sum_{\gamv} V(\gamv) L_1(s_1\mid r_1\gamma_1) \cdots L_n(s_n\mid r_n\gamma_n),
  \]
  can be obtained in the form (\ref{eq:qcorr}), with a suitable state
  and measurement. Namely,
  \[
    \rho = \sum_{\gamv} V(\gamv) \proj{\gamma_1} \otimes \cdots \otimes \proj{\gamma_n},
  \]
  and
  \[
    M^{r_i}_{s_i} = \sum_{\gamma_i} L_i(s_i\mid r_i\gamma_i) \proj{\gamma_i}.
  \]

  Thus we have,
  \[
    \mathtt{LOC}(S\mid R) \subset \text{Q}(S\mid R) \subset \mathtt{BINV}(S\mid R).
  \]
 Bell \cite{bell1964einstein} and Tsilerson \cite{Tsilerson} prove that the above inclusions are strict.
\end{remark}

\subsection{Quantum solutions and quantum correlated equilibria} \label{sec:qeq}

We are now ready to give the definition of a quantum correlated equilibrium.
To start, a \emph{quantum solution} for a game consists of local
measurements $M^{t_i} = (M^{t_i}_{a_i}:a_i\in A_i)$ for player $i$, $t_i\in T_i$,
on a suitable local Hilbert space $\mathcal{H}_i$,
and a state $\rho$ on $\mathcal{H} = \mathcal{H}_1 \ox \cdots \ox \mathcal{H}_n$.
This defines a correlation in canonical form,
\begin{equation}
  \label{eq:quantum-canonical}
  Q(\av\mid\tv) = \tr \rho (M^{t_1}_{a_1} \ox \cdots \ox M^{t_n}_{a_n}),
\end{equation}
and hence expected payoff for player \dio{$i$ that observed type $i$ is}
\[
  \avg{v_{i\dio{,t_i}}(M^{\tv}, \rho)} = \sum_{\dio{\tv_{-i}},\av} P(\dio{\tv_{-i} \mid t_i}) \tr \rho (M^{t_1}_{a_1} \ox \cdots \ox M^{t_n}_{a_n}) v_i(\tv,\av).
\]

In strategic form, the ``quantum correlated'' canonical game goes as follows.
A mediator, who does not know the players' types, has a correlation device
that produces the state $\rho$. The players trust the mediator in using such
device correctly. He sends to each player $i$ the corresponding part of the
state in the space $\mathcal H_i$. He also suggests the measurements
$\{M^{t_i}: t_i \in T_i\}$ to use as a strategy. Note that there is no need
for the players to reveal their types to the mediator, just as in the case
of correlated solutions.

The definition of equilibrium is basically the same as before; we want to capture
the idea that no player has an incentive to deviate from the advice
unilaterally.

\begin{definition}[Quantum correlated equilibrium]
  \label{eq:quantum-corr-eq}
  A quantum solution $(M^{\tv},\rho)$ is a \emph{quantum correlated equilibrium},
  if and only if for all players $i$, \dio{all types $t_i$} and any measurements $N^{t_i} = (N^{t_i}_{a_i} : a_i\in A_i)$,
  \[\begin{split}
    \avg{v_{i\dio{,t_i}}(M^{\tv}, \rho)} &= \sum_{\dio{\tv_{-i}},\av} P(\dio{\tv_{-i} \mid t_i}) \tr \rho (M^{t_1}_{a_1} \ox \cdots \ox M^{t_n}_{a_n}) v_i(\tv,\av) \\
           &\geq \sum_{\dio{\tv_{-i}},\av} P(\dio{\tv_{-i} \mid t_i}) \tr \rho (M^{t_1}_{a_1} \ox \cdots \ox M^{t_{i-1}}_{a_{i-1}}
                                                 \ox N^{t_i}_{a_i}
                                                 \ox M^{t_{i+1}}_{a_{i+1}} \ox \cdots \ox M^{t_n}_{a_n}) v_i(\tv,\av).
  \end{split}\]
  Under the same philosophy as in the sections on belief-invariant and correlated
  equilibria, we then call the canonical solution $Q$ as in eq.~(\ref{eq:quantum-canonical})
  a \emph{canonical quantum correlated equilibrium}, the set of which
  is denoted $\QUEQ(G)$.
\end{definition}

%

\subsubsection{Fitting the quantum model in our framework} \label{sec:qframework} 

It may seem as if we have left the formalism of communication and belief-invariant
developed in the previous section, and of course that is necessarily the
case since we want to talk about quantum correlations. 
Indeed, there is a fundamental difference between quantum equilibria and the belief-invariant equilibria of  Definition 6. Consider a quantum solution as in Definition \ref{eq:quantum-corr-eq}, where a shared state and local measurements implement a canonical solution $Q$. Even if the correlation $Q$ allows only finitely many inputs from player $i$, in the quantum setting this player must be still permitted to perform any of the infinitely many measurements that are physically allowed.
Roughly speaking, this means that whereas we can simulate belief invariant equilibria with a mediator that takes one of the finitely many inputs from each player, implements the correlation $Q$ and returns the action to each player,
for quantum equilibria we need the mediator to receive one among infinitely many measurements (even if there are only finite inputs) from each player, implement the correlation embedded in the quantum state, and returns the action to each player.

It is however possible,
although at a price, to present quantum solutions and quantum correlated
equilibria in our general framework.
To this end, note that in Definition~\ref{eq:quantum-corr-eq} above we have to
consider any one player varying their measurement. Thus, define $R_i := M(\mathcal{H}_i,A_i)$ to be
the set of all possible measurements on $\mathcal{H}_i$ with outcomes
in $A_i$; this is of course an infinite set, in fact it has the structure of
a manifold, but let us not worry about that.
In this way, each $r_i \in R_i$ specifies precisely a measurement and each
possible measurement is represented. Denote this (very big) correlation
$\mathfrak{Q}(\av\mid\rv)$. By definition, we get the following:

\begin{proposition}
  \label{prop:quantum-corr-eq:alt}
  A quantum solution $(M^{\rv},\rho)$ is a quantum correlated equilibrium
  if and only if
  $(\fv,\id,\mathfrak{Q})$ is a (belief-invariant) communication equilibrium,
  where $f_i(t_i) = r_i := M^{t_i} \in R_i$ is defined uniquely by the requirement
  that $r_i$ labels the measurement used in the quantum solution.

  Furthermore, the canonical forms coincide:
  $\widehat{\mathfrak{Q}}  = Q$.
\end{proposition}

The class of quantum correlated equilibria contains the correlated equilibria of
Definition \ref{def:correlated} as a special case:
\begin{proposition}
  \label{prop:classical-quantum-equilib}
  Every correlated equilibrium is a quantum correlated equilibrium.
  Indeed, if $(\gv,Q)$ is a correlated equilibrium, and we define
  the state $\rho$ and measurements $M^{t_i} = (M^{t_i}_{a_i}:a_i\in A_i)$
  as in Remark~\ref{rem:classical-quantum}:
  \begin{align*}
    \rho          &= \sum_{\sv} Q(\sv) \proj{s_1} \ox \cdots \ox \proj{s_n}, \\
     M^{t_i}_{a_i} &= \sum_{s_i} \d_{g_i(t_i, s_i),a_i} \proj{s_i},
  \end{align*}
  where $\d$ is the Kronecker delta function, then $(M^{\rv},\rho)$ is a quantum correlated equilibrium
  which has the same canonical representative $\widehat{Q}$ as $(\gv,Q)$;
  in particular they have the same outcome.
\end{proposition}
\begin{proof}
The state $\rho$ is a mixture of classical advice. For all $i, t_i$ every measurement can also be simulated classically, and locally, with the use of private randomness. Since $(\gv,Q)$ is a correlated equilibrium, no deviation from the suggested measurement can be beneficial to any player.
This shows that $(M^{\rv},\rho)$ is an equilibrium, and by \eqref{eq:quantum-canonical} one can verify that its canonical is $\hat Q$.
\end{proof}

\subsubsection{Discussion and historical background}
From the definition and the observations made above, it follows that
when comparing equilibrium classes at the level of their canonical representatives, quantum correlated equilibria are sandwiched between correlated and belief-invariant ones:
\[
  \CORREQ(G) \subset \QUEQ(G) \subset \BIEQ(G).
\]
There are games where the inclusions are strict. A famous one is the {\tt CHSH} game, which we have used extensively as an example throughout this paper. This is a non-local game, therefore a full coordination Bayesian game.
It follows from \cite{CHSH} that there is a quantum equilibrium which is not in $\CORREQ(${\tt CHSH}$)$, and \cite{Tsilerson} proved that the belief-invariant equilibrium given by \eqref{eq:chsh_ns_strategy} is not in $\QUEQ(${\tt CHSH}$)$.
Both results, though elementary mathematically, constituted breakthroughs in the foundations of quantum mechanics.

The study of quantum correlations appears in many works in physics and computer science (see, for example, the surveys \cite{avis2008bell,palazuelos2015random} and the references therein). On the other hand, there have been several approaches for the use of quantum correlations in game theory (as illustrated extensively in the survey \cite{Guo08}). The general connection with Bayesian games has been made explicit in works like \cite{LaMura03,Brunner13,Pappa15,Brandenburger20150096,lehrer2013garbling}. In particular, \cite{LaMura03,Pappa15} figured out that there could be a quantum advantage even in the case of conflicting interest, \emph{i.e.,} non full-coordination games. This represented a radical shift from the traditional approach in physics: initially, full coordination games were used as a tool to exhibit the difference between quantum and classical behaviours. Physicist used to prove that Nature is not classical by performing a seemingly impossible collaborative task between two or more space-like separated experimenters. Now, quantum correlations can be used to mediate situations of conflict between selfish players, \emph{i.e.,} quantum effects that naturally occur in the microscopic world are used to influence decisions in the macroscopic world. This is the phenomenon we underline here. 
Below we show that the result in \cite{Pappa15} can be extended to $n$ players.

\subsection{Implementing the best belief-invariant equilibrium for {\tt GHZ} with quantum correlations}
\label{sec:quantumGHZ}

Consider the game of Section \ref{sec:ghz}. The correlation $Q$ described in \eqref{eq:ghz_ns_strategy} can be implemented as a quantum
correlated equilibrium, therefore it does not need an informed mediator.
To see this, we will now exhibit one of the constructions of the quantum state
and the measurements that produce the correlation.

Fix the computational basis to $\ket{0} = \begin{bmatrix}1\\0\end{bmatrix} $ and $ \ket{1} = \begin{bmatrix}0\\1\end{bmatrix}$ and consider the following abbreviation for the tensor product: $\ket{xyz} = \ket{x} \otimes \ket{y} \otimes \ket{z}$.

The state is then constructed as follows.
Each player $i$ holds a 2-dimensional quantum system with Hilbert space $\mathcal{H}_i$. Start from the following vector living in $\mathcal{H} = \mathcal{H}_1 \ox  \mathcal{H}_2 \ox \mathcal{H}_3$
$$\ket{\psi} = \frac12 (\ket{111} - \ket{001} - \ket{010} - \ket{100}), $$
and obtain the corresponding state on $\mathcal{H}$:
$$ \rho = \ket{\psi}\bra{\psi}.$$
The measurement operators are the following. For each player $i$, we have
\begin{align*}
\left\{ M^{0}_{0} = \frac{1}{2}  \left( \begin{array}{cc}
1 & 1  \\
1 & 1  \end{array} \right),  \quad
M^{0}_{1} = \frac{1}{2} \left( \begin{array}{cc}
1 & -1  \\
-1 & 1  \end{array} \right)\right \}~~~~ \mbox{(on ~input~0~),}\\
\left\{ M^{1}_{0} =  \left( \begin{array}{cc}
1 & 0  \\
0 & 0  \end{array} \right),  \quad
M^{1}_{1} = \left( \begin{array}{cc}
0 & 0  \\
0 & 1  \end{array} \right)\right\}~~~~~~~~~~~~~\mbox{(on ~input~1~)} .
\end{align*}

It can be checked via eq.\eqref{eq:qcorr} that the above state and measurements
produce the claimed correlation $Q$.
\begin{proposition} \label{prop:quantum_equilibrium_GHZ}
\textit{The quantum strategy $(M^\tv,\rho)$ is a quantum correlated equilibrium.}
\end{proposition}

\begin{proof}
To prove the claim, we fix the measurements of any two parties and then show that the third party by changing his measurements cannot increase his average payoff. Since the shared state is symmetric under the permutation of parties, it is sufficient to consider the case in which measurements of party-1 are variable and measurements of parties 2 and 3 are fixed to $M^{i}_{j}$, where $i,j \in\{0,1\}$. Measurements of party-1 can be expressed as $X^0_0+X^0_1=I$ (on input 0), and $X^1_0+X^1_1=I$ (on input 1). The general operators defining these measurements are:
\begin{align*}
&X^0_0=\frac{1}{2}\left( \begin{matrix} \alpha+a_3&a_1-i a_2\\a_1+i a_2&\alpha-a_3\end{matrix} \right),\\
&\mbox{where}~ (\alpha,\vec{a}) \in \mathbb{R}^{4},~\vec{a}=(a_1,a_2,a_3),~\mbox{and}~\|\vec{a}\|\leq \alpha \leq 2-\|\vec{a}\|;
\end{align*}
\begin{align*}
&X^1_0= \frac{1}{2}\left( \begin{matrix} \beta+b_3&b_1-i b_2\\b_1+i b_2&\beta-b_3\end{matrix} \right),\\
&\mbox{where}~(\beta,\vec{b}) \in \mathbb{R}^{4},~\vec{b}=(b_1,b_2,b_3),~\mbox{and}~\|\vec{b}\|\leq \beta \leq 2-\|\vec{b}\|.
\end{align*}
Note that the constraints on the parameters in the above equations implies that $\|\vec{a}\|\leq 1$ and $\|\vec{b}\| \leq 1$.

The expected payoff of the $1$-st party, \dio{when his type is $t_1$} is given by,
\begin{equation}
\langle v_{1\dio{,t_1}}\dio{(X^{t_1},M^{\tv_{-1}},\rho)}\rangle= \sum_{\dio{\tv_{-1},\av}}\dio{P(\tv_{-1} \mid t_1)}\cdot \mbox{Tr}\{\rho(X_{a_1}^{t_1}\otimes M_{a_2}^{t_2}\otimes M_{a_3}^{t_3})\}\cdot v_1(\tv,\av). \label{3_party_payoff}
\end{equation}
On substituting the values for probability distribution of inputs, quantum probabilities, and utilities,  in eq.(\ref{3_party_payoff}), and simplifying, we obtain that:
\begin{align*}
\langle v_{1,0}\dio{(X^{t_1},M^{\tv_{-1}},\rho)}\rangle&= \frac{1}{12}\{2-(1-\eps)\alpha+(1+\eps)a_1\},\\
\langle v_{1,1}\dio{(X^{t_1},M^{\tv_{-1}},\rho)}\rangle&= \frac{1}{12}\{1+3\eps+(1-\eps)\beta+2(1+\eps)b_3\}.
\end{align*}

Now, using $0\leq \eps \leq 1$, and constraints on parameters $(\alpha, \vec{a})$ and $(\beta, \vec{b})$, it follows that
\begin{align*}
\langle v_{1,0}\dio{(X^{t_1},M^{\tv_{-1}},\rho)}\rangle&=\frac{1}{12}\left\{2-(1-\eps)\alpha+(1+\eps)a_1\right\}\\
&\leq \frac{1}{12}\left\{2-(1-\eps)\|\vec{a}\|+(1+\eps)\|\vec{a}\|\right\}\\
&= \frac{1}{6}\left\{1+\eps\|\vec{a}\|\right\}\leq \frac{1}{6}\left\{1+\eps\right\}.
 \end{align*}
Therefore, $\max_{\alpha,\vec{a}}\{\langle v_{1,0}\dio{(X^{t_1},M^{\tv_{-1}},\rho)}\rangle\}=\frac{1}{6}\left\{1+\eps\right\}$ and this upper bound is achieved for $\alpha=1, \vec{a}=(1,0,0)$, i.e., when $X^{0}_0 = M^0_0$. Similarly,
\begin{align*}
\langle v_{1,1}\dio{(X^{t_1},M^{\tv_{-1}},\rho)}\rangle&=\frac{1}{12}\left\{1+3\eps+(1-\eps)\beta+2(1+\eps)b_3\right\}\\
&\leq \frac{1}{12}\left\{1+3\eps+(1-\eps)\left(2-\|\vec{b}\|\right)+2(1+\eps)\|\vec{b}\|\right\}\\
&= \frac{1}{12}\left\{3+\eps +(1+3\eps)\|\vec{b}\|\right\}\leq \frac{1}{3}\left\{1+\eps\right\}.
\end{align*}
Therefore, $\max_{\beta,\vec{b}}\{\langle v_{1,1}\dio{(X^{t_1},M^{\tv_{-1}},\rho)}\rangle\}=\frac{1}{3}\left\{1+\eps\right\}$, and this upper bound is achieved for $\beta=1, \vec{b}=(0,0,1)$, i.e., when $X^{1}_0 = M^1_0$.

Finally, due to symmetry in the shared state, we get same results when varying the measurements of some other party by keeping fixed the measurements of the remaining two parties. This shows that every player achieves the maximum expected payoff for each type with the suggested measurement, hence the considered strategy is a quantum equilibrium.
\end{proof}
We remark that a similar approach (modifying the game in \cite{Greenberger90}) has been used independently in \cite{situ2015quantum} to obtain a 3-player quantum game with conflict of interest.

\subsection{A game with conflict of interest where quantum correlations achieve optimal social welfare }

We now introduce a variant of the three player game of Section \ref{sec:ghz}, with a distribution of types: $P(0,0,1)=P(0,1,0)=P(1,0,0)=P(1,1,1)=1/4$, and with the following payoff table (with $0\le \eps\leq 1$ and $\d = \frac{2 + \eps}{3}$):
%

\begin{figure}[htb]
\centering
\begin{minipage}{0.5\textwidth}
\begin{game}{2}{2}[0]
    & 0    & 1 \\
0   &$\d$,$\d$,$\d$   &0,0,0\\
1   &0,0,0   &$\eps$,1,1
\end{game}
\qquad
\begin{game}{2}{2}[1]
    & 0    & 1 \\
0   &0,0,0    &1,$\eps$,1\\
1   &1,1,$\eps$    &0,0,0
\end{game}
\subcaption{$\tau = 0$}
\end{minipage}
\begin{minipage}{0.5\textwidth}
\begin{game}{2}{2}[0]
    & 0    & 1 \\
0   &0,0,0   &1,1,$\eps$\\
1   &1,$\eps$,1   &0,0,0
\end{game}
\qquad
\begin{game}{2}{2}[1]
    & 0    & 1 \\
0   &$\eps$,1,1    &0,0,0\\
1   &0,0,0    &$\d$,$\d$,$\d$
\end{game}
\subcaption{$\tau = 1$}
\end{minipage}
\caption{Another modified {\tt GHZ} game. The figure is structured like Figure \ref{tab:GHZ}.}
\label{tab:GHZ_2}
\end{figure}

Clearly, players have conflict of interest in this modified GHZ game. This game has a very interesting connection with the 3-party Mermin inequality \cite{mermin1990}:
\begin{align}
|\langle A_0B_0C_1 \rangle & +\langle A_0B_1C_0 \rangle+\langle A_1B_0C_0 \rangle-\langle A_1B_1C_1 \rangle| \leq 2 \label{mermin}\\
&\mbox{where},~\mbox{for}~i,j,k \in \{0,1\}, \mbox{the random variables } A_{i},B_{j},C_{k}~\mbox{take~value~}\pm 1.\nonumber
\end{align}
In game theoretical terminology, it can be thought that $i,j,k$ are types of the three players and $A_i,B_j,C_k \in\{\pm 1\}$ correspond to the respective actions.
For sake of readability, we do the following relabelling of the actions: $-1\mapsto 0$ and  $+1\mapsto 1$. Moreover, given a generic correlation $Q$ we will set $Q^{ijk}_{abc} = Q(a,b,c \mid i,j,k)$, for types $i,j,k$ and $a, b, c \in \{0,1\}$. Then the expected values of the product of the outcomes are:
%
\begin{equation}
\label{exp}
 \begin{aligned}
  \langle A_iB_jC_k\rangle &= \{Q^{ijk}_{111}+Q^{ijk}_{100}+Q^{ijk}_{010}+Q^{ijk}_{001}\}\\
&-\{Q^{ijk}_{000}+Q^{ijk}_{011}+Q^{ijk}_{101}+Q^{ijk}_{110}\}
 \end{aligned}
\end{equation}
On substituting (\ref{exp}) in (\ref{mermin}), using normalization condition for probabilities, and rearranging, we get,
\begin{align*}
~&|\left\{1-2(Q^{001}_{000}+Q^{001}_{011}+Q^{001}_{101}+Q^{001}_{110})\right\} \nonumber \\
+&\left\{1-2(Q^{010}_{000}+Q^{010}_{011}+Q^{010}_{101}+Q^{010}_{110})\right\} \nonumber \\
+&\left\{1-2(Q^{100}_{000}+Q^{100}_{011}+Q^{100}_{101}+Q^{100}_{110})\right\} \nonumber \\
-&\left\{2(Q^{111}_{001}+Q^{111}_{010}+Q^{111}_{100}+Q^{111}_{111})-1\right\} |\leq 2
\end{align*}
On simplifying we get:
$$
1\leq \mathbb{Q} \leq 3
$$
where,
\begin{align*}
\mathbb{Q}&= (Q^{001}_{000}+Q^{001}_{011}+Q^{001}_{101}+Q^{001}_{110})\nonumber \\
&+(Q^{010}_{000}+Q^{010}_{011}+Q^{010}_{101}+Q^{010}_{110}) \nonumber \\
&+ (Q^{100}_{000}+Q^{100}_{011}+Q^{100}_{101}+Q^{100}_{110}) \nonumber \\
&+(Q^{111}_{001}+Q^{111}_{010}+Q^{111}_{100}+Q^{111}_{111}).
\end{align*}
Let $\mathbb{M}=\langle A_0B_0C_1 \rangle+\langle A_0B_1C_0 \rangle+\langle A_1B_0C_0 \rangle-\langle A_1B_1C_1\rangle$, then $\mathbb{M}=4-2\mathbb{Q}$. The bound we see on quantity $\mathbb{M}$ is derived by considering all possible local correlations, therefore, nonlocal (quantum and belief-invariant) correlations can violate these bounds. However, $\mathbb{M}$ also has algebraic bounds, $-4\leq \mathbb{M}\leq 4$, which is respected by any type of correlation (or probability distribution). The algebraic restriction on $\mathbb{M}$ implies that $0\leq \mathbb{Q}\leq 4$ for any type of correlation. One important feature of Mermin inequality is that its algebraic bounds can be achieved by quantum correlations. We will use this feature of the Mermin inequality to discover interesting properties in the game that we are considering here.

\bigskip For the game of Figure \ref{tab:GHZ_2}, the sum for each type of expected payoffs of the three players in an equilibrium $(\fv, \gv, Q)$ turns out to be:
\begin{equation*}
\sum_{i,t_i}\langle v_{i,t_i}(\fv, \gv, Q)\rangle = \left(\frac{2+\eps}{4}\right)\mathbb{Q}.
\end{equation*}
For any local equilibrium (i.e., correlated or Nash), we have
\begin{equation*}
 \mathbb{Q}\leq 3 ~~\Leftrightarrow~~ \sum_{i,t_i}\langle v_{i,t_i}(\fv, \gv, Q)\rangle \leq \frac{3}{4}(2+\eps).
 \end{equation*}
 For any possible equilibrium (hence, also quantum, belief-invariant or communication) we have:
\begin{equation*}
\mathbb{Q}\leq 4~~\Leftrightarrow~~ \sum_ {i,t_i}\langle v_{i,t_i}(\fv, \gv, Q)\rangle \leq 2+\eps.
\end{equation*}

If the game is played with the same correlation Q defined in \eqref{eq:ghz_ns_strategy} (which we implemented as a quantum strategy $(M^{\tv},\rho)$ in Section \ref{sec:quantumGHZ}), the sum of  the expected payoffs $\langle v_{i,0}(M^{\tv},\rho)\rangle + \langle v_{i,1}(M^{\tv},\rho)\rangle =\frac{2+\eps}{3} ~\forall i\in\{1,2,3\}.$ Moreover, the strategy is a quantum equilibrium (the proof is very similar to the one of Proposition \ref{prop:quantum_equilibrium_GHZ}; only the objective function changes which is again easy to maximize).

Therefore, an interesting feature of the game of Figure \ref{tab:GHZ_2} is that with the considered quantum equilibrium we obtain the optimal fair equilibrium, \emph{i.e.,} no other quantum equilibrium, belief-invariant equilibrium, or even communication equilibrium can do better than our equilibrium in this second GHZ game. This is a new feature which was not revealed in the two party modified CHSH-game considered in \cite{Pappa15} where optimal fair quantum correlated equilibrium was found, however, for the modified CHSH-game belief-invariant equilibrium, and communication equilibrium can do better than the optimal quantum fair equilibrium.

%

\section{Conclusions and open problems}
\label{sec:conclusion}
We have formally introduced the class of belief-invariant communication equilibria and its quantum mechanical version. Even if such classes appeared implicitly in previous work, a systematic study and an hunt for useful applications was not performed before. 
The interested reader will find in the Appendix numerous potential directions for further research.
With this work we would like to open the way for collaboration between the quantum information, the theoretical computer science, and the game theory community, to address the numerous open problems.
We conclude the paper with a list of the ones we could think of.

\begin{enumerate}
\item \textit{Complete the complexity scenario.}
In the Appendix we discuss some computational complexity facts. For example, verifying that a solution is an equilibrium is easy if the number of actions is bounded, while finding the optimal (quantum) correlated equilibrium is a hard task, given the connection with multi-prover interactive proofs. However, how difficult is to sample a quantum or a belief-invariant equilibrium? Are there classes of games where this is easy, like the succinct games of \cite{Papadimitriou08}?

\item \textit{Get large separations and upper bound the largest possible separation.}
Full coordination games are used to design Bell tests, experiments that quantify how different quantum mechanics is from classical physics. For this fundamental task, the quantity of interest is the separation between the largest expected payoff at a quantum and at a correlated equilibrium. The race for large separations was settled in \cite{Buhrman12}, where the authors exhibited a game that almost matches the upper bound proven in \cite{Junge10}.
In our context here, large separations would translate to economical or social convenience of implementing communication equilibria while respecting the privacy of the player.
Are there conflict-of-interest games where the quantum correlated equilibrium leads to a much better social optimum than the correlated one?
Is there an upper bound like the one of \cite{Junge10}?
\item \textit{Can any non-local game be converted in a conflict-of-interest game?} This is a question from \cite{Pappa15}. The non-local games are the above-mentioned coordination games used in physics. It would be interesting if all these games also lead to cases in which a conflict-of-interest situation can be improved with quantum or belief-invariant correlations.
\item \textit{Application to other relevant games.} This is a very natural question. Can belief-invariance be beneficial for scheduling problems, market dynamics, and any other topic of practical interest?
\item \textit{Development of automatic belief-invariant advice on large network games.} We show in Appendix \ref{sec:app} that belief-invariance can be useful in network games. Can we design an automatic system that calculates and distributes belief-invariant advice to large-scale network, in order to reduce the congestion? This does not need to be optimal, and an approximation would already have great practical applications.
\end{enumerate}

\section*{Acknowledgements}
This work is partially supported
by the European Commission (STREP ``RAQUEL''),
by the Spanish MINECO (grants FIS2013-40627-P, FIS2016-80681-P and MTM2014-54240-P) with the support of FEDER funds,
by the Generalitat de Catalunya CIRIT (project 2014-SGR-966),
by the ERC (AdG ``IRQUAT''),
by the Comunidad de Madrid (project QUITEMAD+-CM, S2013/ICE-2801).
This work was made possible through the support of grant \#48322 from the John Templeton Foundation. The opinions expressed in this publication are those of the authors and do not necessarily reflect the views of the John Templeton Foundation.
VA and DF are supported by GNCS - INdAM.

The authors thank Andris Ambainis, Oihane Gallo Fern\'andez, Boris Ginzburg, Dmitry Kravchenko, Giuseppe Persiano, Mikl\'os Pint\'er, Laura Santucci, Johannes Schneider and Ronald de Wolf
for useful discussions that improved the accessibility and the content of the paper.
Furthermore, we acknowledge interesting conversations with Ignacio Villanueva, Carlos Palazuelos
and David P\'{e}rez Garc\'{i}a on the open problem regarding large separations between quantum
and classical advice in relation to the degree of competitiveness of the game.
We also thank the anonymous referees for many useful suggestions that improved the readability and the scientific content of the paper.

\hyphenation{sta-ble}
\bibliographystyle{alpha}
\bibliography{gametheory}

\appendix
\section{Further discussion}
\label{apx:discuss}

\medskip
We discuss some other properties of the classes of equilibria defined in the main text.

\paragraph{Necessity of an informed mediator.}
The correlated equilibrium class does not need that the mediator is informed about the types of players, since its advices are based on local correlations and therefore the corresponding shared random variables are independent of the types. For some cases the players
can also get rid of the mediator completely, and base their strategy on
the observation of a single independent shared random variable, such as
meteorological data.
This however suggests that all players receive the
same information as advice. There are examples of correlated
equilibria where this is not given and indeed not possible, such as the game of Chicken \cite{rapoport1966game}.
Different schemes to get rid of a mediator based on communication between players have been studied in \cite{ForgesSunspots}.

In contrast, in the case of belief-invariant and communication equilibria,
the players seem to need a trusted mediator to implement the correlation\footnote{In special cases, it is known that mediator can be replaced by distributed devices, such as \emph{cheap talk} \cite{gerardi2004unmediated}.}.
We say ``seem'' because strong experimental evidence from physics suggests
that it is possible to go beyond the local correlations without a mediator
by using quantum effects. However quantum mechanics cannot cover the whole
class of belief-invariant correlations. The correlation \eqref{eq:chsh_ns_strategy}
given below is not achievable in quantum mechanics without a mediator (as proven in \cite{Tsilerson}). Therefore, unless quantum mechanics
is falsified in the future and replaced by another theory, a mediator
is needed to implement the complete belief-invariant class. The quantum mechanical
class will be discussed later in Section \ref{sec:quantum}.

\paragraph{Privacy of the players.}
Clearly, in order to implement a correlated equilibrium, nobody else
other than player $i$ needs to learn the type $t_i$.

The class of quantum correlations can be completely realized only with the use
of quantum information processing. Such technology is developing rapidly and
it is already available to experimental physicists, as reported in \cite{Pappa15}.
Physicists often imagine ideal devices called ``non-signalling boxes'' that implement
all the class of belief-invariant correlations without revealing the types to
a mediator. However, quantum mechanics is the best-known theory to describe our
reality, and there is no known super-quantum theory that allows the existence of
the non-signalling boxes. Therefore, to the best of our knowledge of nature,
the quantum class is the best feasible way of obtaining correlations without
revealing players' types to a mediator.

The belief-invariant
class allows for more correlations at the expense that a trusted mediator
might learn something about the types. The use of a belief-invariant
correlation guarantees however that the mediator will be the only one
learning the types and no player $j$ will learn $t_i$.
It is not always possible to respect this requirement in the more general
class of communication equilibria.

Note that our concept of privacy in correlated and belief-invariant equilibria
is much stronger than the well-known concept of differential privacy~\cite{pai2013privacy}.
Indeed, we say that an equilibrium is private in an information-theoretic sense
and assume that each player cannot obtain any new information about the other players while playing the game;
differential privacy, instead, assumes only that each player cannot obtain more than epsilon information, for small and positive epsilon.
Moreover, differential privacy usually guarantees privacy only when the number of players is large,
while our privacy concept applies to any number of players.

\paragraph{Computational complexity.}
The equilibrium concept and its variations discussed so far are
useful to understand the behaviour of the players.
A fundamental question (see for example \cite{Papadimitriou08}) is how
one could calculate such an equilibrium or even just verify that a given
set of strategies is an equilibrium, or on the other hand find
an equilibrium that optimizes some other parameter, such as a
social payoff.
Not that much is known for the class of equilibria we just discussed. Below we
mention the results we are aware of, and we leave as an interesting open
problem to complete the picture.
Note that, except where otherwise specified, the computational complexity
will be with respect to the size of the $n$-player incomplete information \emph{game specification},
that consists of a list of probabilities $P(\tv)$, for $\tv \in T$,
and a list of the payoff function values $v_i(\tv,\av)$,
one for each player $i=1,\ldots,n$, each $\tv\in T$, and each $\av\in A$.%
\footnote{Note that the size of the game specification is exponential in the number of players.
For this reason, many computational complexity results have been given only for games
that are succinctly representable, that is games that can be fully specified by a number of parameters that
is polynomial in the number of players, types and actions.}

\medskip
{\bf Nash equilibria} of complete information games are hard to find:
in fact, it is known that the problem is PPAD-hard even for two-player games \cite{daskalakis2009complexity,Chen:2009}.
Since games of incomplete information contain complete-information games as a special case, they are at least as hard.
On the other hand, from Definition~\ref{def:communication-simpler}, it turns out that
it is possible to check in polynomial time whether a solution $(\id, \id, Q)$ in canonical form is a Nash equilibrium
whenever the number of actions available to players is bounded.

\medskip
{\bf Correlated equilibria} of complete information games can be found in time that is polynomial
in the size of the game specification\footnote{The game specification of complete information game is exactly the specification of a game of incomplete information with $T$ being a singleton.} through linear programming~\cite{hart1989existence}.
Recall that a game of incomplete information can be modelled as a game of complete information where the players'
strategies are $A_i^{T_i} = \{g_i\colon T_i\rightarrow A_i\}$,
i.e., all possible functions from $T_i$ to $A_i$.
However, this not only disregards the special nature of the payoff functions that depend
only on input-output pairs of the strategy, but exponentially increases the size of the game specification.
Thus, the above result does not extend, and thus we do not know
whether it is possible to find correlated equilibria in polynomial time.

Still, optimal correlated equilibria of incomplete information games are hard to find
(they belong to the complexity class of NP-hard problems) even for full coordination games,
with respect to the size of the game specification.
This can be proven by embedding notoriously hard problems (like the chromatic number of a graph)
into cooperative games (see, for example, the game used in \cite{cameron2007quantum}.

\medskip 
{\bf Quantum equilibria} of complete ad incomplete information games are easy hard to find, 
because of the a connection with multi-prover interactive proofs. 
However, verifying that  $(M^{\tv},\rho)$ is a quantum correlated equilibrium can be
done via semidefinite programming. As argued in \cite{Pappa15} one can fix the
other players' strategies and check that for each type $t_i$ the optimal
strategy of player $i$ is $M^{t_i}$.
Since this must be done for each type profile, the running time is polynomial
in the size of the description of the supposed equilibrium, 
i.e., the dimension of the matrices describing the measurement and the quantum state.

\medskip
{\bf Belief-invariant equilibria} of full coordination games of incomplete
information can be instead found in time that is polynomial in the size of game description
via linear programming. This is because the set of non-signalling correlations is
defined by polynomially many non-negative variables subject to polynomially many
linear inequalities.
(See, for example, the LP in \cite[page 8]{buhrman2014parallel}.)
We do not know yet if this extends to conflict of interest games, and this
is one of the major open problems in the present theory. Our intuition is
that the belief-invariant equilibria of games of incomplete information
are the ``right'' analogue of correlated equilibria of games of complete
information; from this perspective one might speculate that belief-invariant
equilibria can be found efficiently, such as~\cite{hart1989existence}.

\section{Potential applications}
\label{sec:app}
In the previous sections we introduced the concepts of correlated and belief-invariant equilibria,
and briefly highlighted some of their properties, mainly with respect to privacy, computational complexity and social welfare maximization.
We saw that very little is known about these equilibria, and thus these equilibrium concepts can stimulate theoretical research along these and other directions.

However, our interest in these equilibrium concepts is not only of theoretical nature.
Indeed, it turns out that privacy and high social welfare are desirable properties in many real world settings.
For these settings, implementing a good correlated or a good belief-invariant equilibrium can help the players
to reach better equilibria guaranteeing information-theoretic privacy.
We list below examples of real-world settings in which the concepts of correlated and belief-invariant equilibria can be relevant.

\paragraph{Trade secrets in markets.}
Suppose two or more companies are in competition for a share of the market (e.g., Coke and Pepsi, or Microsoft and Apple).
Each company is trying to introduce a new product, and the features of the products are trade secrets.
They have an incentive to cooperate in order to minimize the production costs but, at the same time, they have an incentive \emph{not} to cooperate, as the negotiation might expose their secrets.
In such a setting, if there were a trusted third party able to implement a correlated or a belief-invariant equilibrium,
this would not longer be an issue.

Note that this simple example can be generalized to every setting
in which players' payoff are affected by the ability of ``guessing the other players' types''.

\paragraph{Advertising.}
A typical setting in which the ability of guessing the other players' type is of huge importance is advertising (see, for example, \cite{CummingsLPR15}).
In such a setting an advertiser has a product to advertise, whose absolute quality is unknown to the potential users.
The advertiser can adopt different advertising strategies (e.g., viral advertising, commercials, newspaper ads, web ads),
whose success depends on the features of the subject to which the advertisement is aimed.
On the other side, users would like to receive ads only for high-quality products, therefore they do not like to reveal their interests.

Hence, as above, it would be useful both for advertiser and for users to correlate their actions,
so that the advertisers would be able to make more successful advertising campaign
and the users could receive more ads for high-quality products than for low-quality ones.
However, because of the privacy issues discussed above,
this correlation may be effectively implemented only if there is a trusted mediator
that is able to find a good correlated or belief-invariant equilibrium.

\paragraph{Coalition formation}
In coalition formation problems, $n$ players want to arrange coalitions, \textit{i.e.}, a partition of the set of players. (As a simple example, one can imagine children organizing the two teams before a friendly football match.) Informally, a coalition formation setting relevant to us is as follows: each player has a list of desirable allies, and communicates to a mediator a ordered list of preferred subsets of the players in which they would like to be included. The mediator, then, announces the partition trying to create an equilibrium that maximizes player's happiness with the choice.
In such a setting, a belief-invariant correlated equilibrium could help designing solutions where the players' real preferences are kept secret. For recent developments on coalition formation, see
\cite{Gallo16} and the references therein.

\paragraph{Network congestion.}
Another example in which correlated and belief-invariant equilibria may be useful is for network congestion.
Suppose that the routes taken by people driving during rush hours were correlated in some way. (An example could be a GPS device or smartphone application, on which people select their starting and destination points and receive a suggested route.)
Can this correlation reduce the congestion of the network?
As we show next, not only this is possible,
but this correlation can also be implemented privately,
so that the suggestion does not reveal the sources or the destinations
of other players.
Note that this may be required in order to avoid privacy leakages.

Specifically, we show an application of belief-invariant equilibria in \emph{network congestion games (aka selfish routing)}.
Here, a network is modelled as a labelled graph $G$, defined by a set $V$ of vertices,
a set of edges $E \subset V\times V$ and for each $(u,v) \in E$
a cost function $c_{uv} \colon \mathbb{N} \rightarrow \mathbb{R}$.
A network congestion game goes as follows. Each player is associated with a source and a target node
and they have to decide the route to take.
The strategy set for a player contains all possible paths from source to target,
and the utility for a player is \textit{minus} the sum of the cost functions in the edges of the chosen path.
The edges get ``congested'' as a function of the number of players using them, i.e.,
if $x$ players choose edge $(u,v)$, then each of them faces cost $c_{uv}(x)$.

We consider an ``incomplete information'' version of network congestion games,
where the source-targets are decided by the players' types.
Then, by tweaking the {\tt CHSH} game we can exhibit an instance of selfish routing
that demonstrate that belief-invariant equilibria can help reducing the social cost.Our example is illustrated in Figure \ref{fig:CHSH_network}.

\begin{figure}[htb]
\begin{center}
\begin{tikzpicture}
\tikzstyle{new}=[circle,fill=White,draw=Black,minimum width=22pt]
		\node [style=new] (0) at (0, 0) {$s_1$};
		\node [style=new] (1) at (0, 2) {$s_0$};
		\node [style=new] (2) at (3, 0) { $d$ };
		\node [style=new] (3) at (3, 2) { $u$};
		\node [style=new] (4) at (5, 1) { $t$ };
		\draw (0) to node[above, near start ]{$x$} (2);
		\draw (0) to node[above, near start ]{$x$} (3);
		\draw (1) to node[above, near start ]{$1$} (3);
		\draw (1) to node[above, near start ]{$1$} (2);
		\draw (2) to node[above, near start ]{$1/x$} (4);
		\draw (3) to node[above]{$1/x$} (4);
\end{tikzpicture}
\caption{A network congestion game based on {\tt CHSH}. The edge labels are the cost functions of the edges, where $x$ is the number of players using that edge.}
\label{fig:CHSH_network}
\end{center}
\end{figure}

The game goes as follows: there are two players, with binary types selected uniformly at random.
The source-target pairs are chosen as follows:
\begin{itemize}
\item if $t_i = 0$ then Player $i$ starts at vertex $s_0$ and wants to reach $t$,
\item if $t_i = 1$ then Player $i$ starts at vertex $s_1$ and wants to reach $t$.
\end{itemize}

\newcommand{\up}{\mathtt{UP}}
\newcommand{\down}{\mathtt{DOWN}}

The strategies can be summarized in two meaningful choices: from his source point, a player can decide to go \texttt{UP} towards the vertex $u$ or \texttt{DOWN} towards the vertex $d$ and from there take final step towards $t$.

\renewcommand{\gamestretch}{1.5}
\begin{figure}[htb]
\centering
\begin{game}{2}{2}[$t_1\cdot t_2=0$]
    & \texttt{UP}    & \texttt{DOWN} \\
\texttt{UP}   &$\frac{3}{2}$,$\frac{3}{2}$   &2,2\\
\texttt{DOWN}   &2,2   &$\frac{3}{2}$,$\frac{3}{2}$
\end{game}
\qquad \qquad
\begin{game}{2}{2}[$t_1 \cdot t_2=1$]
    & \texttt{UP}    & \texttt{DOWN} \\
\texttt{UP}   &$\frac{5}{2}$,$\frac{5}{2}$    &$\frac{3}{2}$,$\frac{3}{2}$\\
\texttt{DOWN}   &$\frac{3}{2}$,$\frac{3}{2}$    &$\frac{5}{2}$,$\frac{5}{2}$
\end{game}
\caption{Cost matrix for the network congestion {\tt CHSH} game}
\label{fig:CHSH_network_payoffs}
\end{figure}
\renewcommand{\gamestretch}{1}

The cost table is in Figure \ref{fig:CHSH_network_payoffs}.
It is clear from the table that the situation is similar to Figure~\ref{fig:CHSH}.
Players with pure strategies can reach an equilibrium by choosing $(\up,\up)$ or $(\down,\down)$ regardless of their types,
and have an expected cost of $\frac{3}{4}\cdot \frac{3}{2} + \frac{1}{4} \cdot 2 = \frac{13}{8}$.
A belief-invariant correlation as in \eqref{eq:chsh_ns_strategy} gives an equilibrium with expected cost of $\frac{3}{2}$..

Note that one can introduce conflict of interest into the game above by modifying the network in Figure \ref{fig:CHSH_network} as follows:
\begin{figure}[htb]
\centering
\begin{tikzpicture}
\tikzstyle{new}=[circle,fill=White,draw=Black,minimum width=22pt]
		\node [style=new] (0) at (0, 0) {$s_1$};
		\node [style=new] (1) at (0, 2) {$s_0$};
		\node [style=new] (2) at (3, 0) { $d$ };
		\node [style=new] (3) at (3, 2) { $u$};
		\node [style=new] (4) at (5, 0) { $t'$ };
		\node [style=new] (5) at (5, 2) { $t''$ };
		\draw (0) to node[above, near start ]{$x$} (2);
		\draw (0) to node[above, near start ]{$x$} (3);
		\draw (1) to node[above, near start ]{$1$} (3);
		\draw (1) to node[above, near start ]{$1$} (2);
		\draw (3) to node[above, near start ]{$1/x$} (4);
		\draw (2) to node[below=0.15cm, near start ]{$1/x$} (5);
		\draw (4) to node[right]{$\eps$} (5);
\end{tikzpicture}
\caption{A network congestion game based on {\tt CHSH}. The edge labels are the cost functions of the edges, where $x$ is the number of players using that edge.}
\label{fig:CHSH_network2}
\end{figure}

In this graph, whatever the types are, always assign Player 1 to target $t'$ and Player 2 to target $t''$.
Now the fist player will prefer the strategy $\up$ and the second player will prefer the strategy $\down$,
in a situation similar to the game in Figure~\ref{fig:CHSH2}.

\end{document}